\documentclass{amsart}

\usepackage{amsmath,amssymb}
\usepackage{amsfonts}

\usepackage{graphicx}
\usepackage{appendix}
\usepackage{amsaddr}
\title[Modelling allocation match under uncertainty]{Modelling the expected probability of correct assignment under uncertainty}
\author{Tom Dvir$^{1,4}$, Renana Peres$^2$*, and Ze\'ev Rudnick$^3$}
\address{1 Racah Institute of Physics, The Hebrew University, Jerusalem 91904, Israel}
\address{2 School of Business Administration, Hebrew University of Jerusalem, Jerusalem 91905, Israel}
\address{3 School of Mathematical Sciences,
Tel Aviv University, Tel Aviv 69978, Israel}
\address{4 QuTech and Kavli Institute of Nanoscience, Delft University of Technology, 2600 GA Delft, The Netherlands}
\address{* corresponding author: renana.peres@mail.huji.ac.il}

\begin{document}

\newtheorem{open}{Open Problem}
\newtheorem{thm}{Theorem}%[section]
\newtheorem{lem}[thm]{Lemma}
\newtheorem{prop}[thm]{Proposition}
\newtheorem{cor}[thm]{Corollary}
\newtheorem{conj}{Conjecture}%[section]
\newtheorem{exercise}{Exercise}
\newtheorem{example}{Example}
% \newtheorem*{example*}{Example}

%\numberwithin{equation}{section}

\renewcommand{\figurename}{Supplementary Figure}  % this renames Figures to "Supplementary Figures"

\newcommand{\Z}{{\mathbb Z}} %cph changed from \mathbf
\newcommand{\Q}{{\mathbb Q}}
\newcommand{\R}{{\mathbb R}}
\newcommand{\C}{{\mathbb C}}
\newcommand{\N}{{\mathbb N}}
\newcommand{\FF}{{\mathbb F}}
\newcommand{\fq}{\mathbb{F}_q}
\newcommand{\feq}{\overline{\mathbb F}_q}

\newcommand{\rmk}[1]{\footnote{{\bf Comment:} #1}}

\renewcommand{\mod}{\;\operatorname{mod}}
\newcommand{\ord}{\operatorname{ord}}
\newcommand{\TT}{\mathbb{T}}
\renewcommand{\i}{{\mathrm{i}}}
\renewcommand{\d}{{\mathrm{d}}}
\renewcommand{\^}{\widehat}
\newcommand{\HH}{\mathbb H}
\newcommand{\Vol}{\operatorname{vol}}
\newcommand{\area}{\operatorname{area}}
\newcommand{\tr}{\operatorname{tr}}
\newcommand{\norm}{\mathcal N} % norm =(\frac{ n+\sqrt{n^2-4}} 2)^2
\newcommand{\intinf}{\int_{-\infty}^\infty}
\newcommand{\ave}[1]{\left\langle#1\right\rangle} %  average
\newcommand{\Var}{\operatorname{Var}}
\newcommand{\Prob}{\operatorname{Prob}}
\newcommand{\sym}{\operatorname{Sym}}
\newcommand{\disc}{\operatorname{disc}}
\newcommand{\CA}{{\mathcal C}_A}
\newcommand{\cond}{\operatorname{cond}} % conductor
\newcommand{\lcm}{\operatorname{lcm}} 

\newcommand{\Kl}{\operatorname{Kl}} %Kloosterman sum
\newcommand{\leg}[2]{\left( \frac{#1}{#2} \right)}  % Legendre symbol

\newcommand{\sumstar}{\sideset \and^{*} \to \sum}

\newcommand{\LL}{\mathcal L} %L-function of u
\newcommand{\sumf}{\sum^\flat}
\newcommand{\Hgev}{\mathcal H_{2g+2,q}}
\newcommand{\USp}{\operatorname{USp}}
\newcommand{\conv}{*}
\newcommand{\dist} {\operatorname{dist}}
\newcommand{\CF}{c_0} % Fejer constant
\newcommand{\kerp}{\mathcal K}

\newcommand{\fs}{\mathfrak S}
\newcommand{\rest}{\operatorname{Res}} % resultant
\newcommand{\af}{\mathbb A} % affine line
\newcommand{\Li}{\operatorname{Li}}
\newcommand{\Sel}{\mathcal S}
\newcommand{\SF}{\mathbf 1_{\rm SF}}
 \newcommand{\SFz}{\mathbf 1_{{\rm SF},z}}

\newcommand{\T}{\mathbb T} % torus
\newcommand{\E}{\mathbb E}
\newcommand{\length}{\operatorname{length}}
\newcommand{\nattribs}{K} % number of attributes

\newcommand{\attrib}{\mathcal A} % the space of attributes

%\title[The expected probability of correct assignment]{Modelling allocation match under uncertainty: The expected probability of correct assignment}
%\author{Tom Dvir, Renana Peres and Ze\'ev Rudnick}
%\address{Racah Institute of Physics, The Hebrew University, Jerusalem 91904, Israel}
%\email{tom.dvir@gmail.com}
%\address{School of Business Administration, Hebrew University of Jerusalem, Jerusalem 91905, Israel}
%\email{peresren@huji.ac.il}
%\address{School of Mathematical Sciences,Tel Aviv University, Tel Aviv 69978, Israel}
%\email{rudnick@tauex.tau.ac.il}

%\thanks{Z.R. was supported  by an Advanced Grant from the European Research Council under the European Union's Horizon 2020 research and innovation programme/ERC grant agreement n$^{\text{o}}$~786758  }

 %\date{\today}

\maketitle

\begin{abstract}
When making important decisions such as choosing health insurance or a school, people are often uncertain what levels of attributes will suit their true preference. After choice, they might realize that their uncertainty resulted in a mismatch: choosing a sub-optimal alternative, while another available alternative better matches their needs.

We study here the overall impact, from a central planner's perspective, of decisions under such uncertainty. We use the representation of Voronoi tessellations to locate all individuals and alternatives in an attribute space. We provide an expression for the probability of correct match, and calculate, analytically and numerically, the average percentage of matches. We test dependence on the level of uncertainty and location.

We find overall considerable mismatch even for low uncertainty - a possible concern for policy makers. We further explore a commonly used practice - allocating service representatives to assist individuals' decisions. We show that within a given budget and uncertainty level, the effective allocation is for individuals who are close to the boundary between several Voronoi cells, but are not right on the boundary.
\end{abstract}

\section*{Introduction}

Important decisions people make, such as choosing health insurance, or choosing a school, require complex considerations. In many cases these considerations are further complicated by the uncertainty, or error of individuals in understanding what levels of specific attributes match their true preferences.
For instance, in  choosing a school, people might find it hard to specify what a "good" school means to them in terms of the level of specific attributes such as the number of Math hours, intensiveness of the music program, vocational training, geographic location, or whether the athletics program should include Quidditch (Jenkin 2015) \cite{Jenkin}. After their children start attending the school, they might realize that they find the Math program less demanding, the music program too intensive, or that a 20 minute walk to school is more strenuous than they expected. Thus, they would have been happier with a school that has slightly different values on these attributes. Such patterns of post choice evaluation, regret, and disappointment have been empirically documented in the past literature (e.g. Westbrook 1987 \cite{W}; Inman, Dyer and Jia 1997 \cite{IDJ}). Representing the relevant domain in the attribute space we say that while individuals might claim to know where their preferences are located in the space, there is often uncertainty as to their true desired location. Only after the choice, they might realize that their perceived location doesn't match their true needs and desires. This uncertainty might be a result of insufficient information about the meaning of different levels of attributes for them (e.g. what parental involvement, or an intensive music program require from them), or misconception as to what they really want. 

In a market with several alternatives, such uncertainty might result in a mismatch - that is, choosing an alternative that is sub-optimal, although there are other available alternatives which better match one's real needs. For example, while parents might be certain that they want the school with the intensive Math program, they might have actually been better off in a school with a less intensive program. Therefore their true preference would be in a slightly different location in the attribute space than what they initially thought they were. The choice literature indicates that mismatches happen when the choice task is complicated, or when individuals do not have enough previous experience with the specific choice task (Mosteller and Nogee 1951) \cite{MN} 

Considerable effort is invested in reducing uncertainty to avoid mismatch in important decisions. Financial planners are used to consult in choosing health plans (McClanahan 2014) \cite{Mc}, and advisors assist in pension plan choice (PFau 2016) \cite{PFau}.   Residents of major cities such as New York City employ expensive private consultants to assist in choosing a school (Harris and Fessenden 2017) \cite{HF}. 
From the perspective of the central planner that provides and supervises these services, too many mismatches are undesirable. A large group of dissatisfied service recipients might cause a decrease in the overall social welfare, which, in turn might lead to social and economic consequences. Assuming that the central planner wants to maximize the social welfare, as a goal by itself or in order to serve political and economic stability, it would better to minimize mismatches.  

Our goal in this paper is to study the overall impact, from the perspective of the central planner, of decisions under the uncertainty described above (which we term hereafter as "uncertainty in preferences"). Similar to the school choice problem (Holmes Erickson 2017 \cite{HE}; Abdulkadiroglu et al. 2020 \cite{Abdul}), the decision scenarios we model apply to 
high involvement, multiple attribute goods and services that are monitored by a central planner. They can be credence/experience goods and services, with a high importance for customer satisfaction and a high chance for post-choice evaluation and regret. While some of their attributes (such as distance or cost) might be very directional (a rational consumer will prefer zero distance and zero cost), many other attributes (e.g. level of religiousness, intensity of the math program, hours of French per week etc.) are a matter of personal preference and can greatly vary between individuals. While the general formulation of the problem can incorporate a large number of market conditions and variables, we wish to work with a restrained set of conditions that will enable us to focus on the effect of uncertainty. Therefore, we focus on the case of no supply constraints, no specific market structure, and no interactions between individuals. Our modeling framework enables expansion to include these scenarios.

We use the representation of Voronoi tessellations to describe an attribute space with different alternatives, each having its attraction basin. 
Individuals can be also located in this space, according to their preference. Each individual has a perceived location, but since individuals might not correctly estimate the attribute levels that match their needs, this perceived location might be distant from their true preference, up to a certain uncertainty factor. The uncertainty creates an error in the perceived location of the individual, and hence can place the individual in the attraction basin of another, sub-optimal alternative, causing a mismatch.

We focus on the probability of correct match - that is, when the choice made is indeed the best alternative for this individual. We provide an expression for the probability for correct match, and show how it depends on the location in the attribute space and on the level of uncertainty. We give a formula for the average percentage of matches for low uncertainty level and use numerical simulation to extend the description for larger uncertainty.

We then extend our model by including a policy to help individuals obtain the correct decision and avoid mismatches. In some cases the central planner might offer "front-desk" services, which provide help through face-to-face or phone meetings. Such services are effective but costly. We use our model to study how  the authority can allocate service representatives to individuals within a given budget in a way that will maximize the overall level of match.

Our contribution is by studying {\bf decisions under uncertainty in preferences} from the perspective of the {\bf central planner}. We draw inspiration from two streams of literature: Decisions under uncertainty, and Matching theory. Decisions under uncertainty have been mostly modeled from the individual's point of view, and focused on the information search of individuals (Branco, Sun and Villas-Boas 2012) \cite{BSV}, on how they sample the choice alternatives (Chick and Frazier 2012) \cite{CF}, how they use social influence to compensate for the missing information (Lopez-Pintado and Watts 2009)\cite{LPW}, and how they update their preferences based on each additional information bit they receive (Erdem and Keane 1996) \cite{EK}.  These models often consider factors such as expected utility from each alternative and risk aversion (Machina 1987) \cite{Ma}. In choice modeling, random utility models were used to describe uncertainty in choice, under the assumption that some attributes are unobserved and are represented as random variables (e.g. Ben-Akiva and Lerman 1985 \cite{BAL}), or, alternatively, that the decision-making individual considers each time only a subset of the attributes (Becker, DeGroot, and Marschak 1951 \cite{BDM}).  Works on post-choice evaluation (e.g. Inman, Dyer and Jia 1997) \cite{IDJ} emphasized factors such as satisfaction and regret. This body of literature focuses on uncertainty in one's understanding of the true value of the suggested alternatives, or, as in the random utility models, that the entire attribute space is not taken into account during the choice. Our focus is on an attribute space and a set of alternatives that are entirely known to the individual, and the uncertainty in one's understanding of his/her own needs and wants. 

The implications of choice from the central planner's perspective have mostly been studied without relating to uncertainty. Studies in matching theory (Gale and Shapley 1961 \cite{GS}; Roth 1986 \cite{Roth}) suggest algorithms for matching between individuals and outlets in various scenarios (schools, houses, hospital residency (see S\"onmez and \"Unver (2011) \cite{SU} for review), where slots are limited, requiring one of the sides or both, to rank their mutual preferences. Recent works on matching have begun to incorporate uncertainty in various forms: Ehlers and Massó (2015) \cite{EM} describe a matching game where players are not sure about the preferences of other players. Hazon et al. (2012) \cite{Hazon} study forecasting voting patterns, where the ranking of candidates for each voter is not fully known to an outside observer. Aziz et al. (2020) \cite{Aziz} study the case where the individuals themselves are not certain in their rankings, but rather rank their preferences with a probability smaller than 1. These models are usually characterized by: (1) assuming limited capacity (otherwise all individuals get what they want); and (2) not having a direct access to the attributes, but rather to a rank ordering of alternatives. Their focus is to find the best matching algorithm that will create stable equilibrium.  

Our modeling perspective draws from both streams - similar to the matching models we deal with matching alternatives to individuals, from the perspective of a central planner. Similar to the decision-under-uncertainty problems, our model deals directly with the attributes and does not use ranking of alternatives. However, we do not focus on the individual level, but rather look at the entire set of alternatives and individuals. We do not assume capacity constraints since, in the presence of uncertainty, mismatches can occur even without capacity limitations. 
Our focus is not empirical estimation, or efficient matching algorithm but rather to measure the probability for correct match and its dependence on various market factors. To the best of our knowledge, this work is the first to suggest a measure for the overall probability of matches, and calculate analytically its average value. Our model enables studying specific policies for minimizing the mismatch, such as the use of service representatives. 

\section*{The space of attributes and Voronoi tessellations}
Our goal is to calculate the overall impact, from the perspective of the central planner, of decisions under uncertainty in preferences. To do so, we want to define a measure for the probability of a correct match for every possible individual preference, and then calculate its average value over a population. As explained above, most matching algorithms (S\"onmez and \"Unver (2011) \cite{SU} assume limited capacity, and the criterion for the optimal overall match is a stable equilibrium – that is, there is no pair of individuals who would be better-off by switching the alternatives they were assigned with. Therefore, these algorithms do not provide a continuous metric for the probability of a correct match. Individual level decision models that incorporated uncertainty (Erdem and Keane 1996 \cite{EK}; Ben-Akiva and Lerman 1985 \cite{BAL}) were used more for empirically estimating one's utility and rarely provide an overall view of all the individuals and alternatives. The representation we seek is one that considers the entire attribute space and range of alternatives, allows representation of the alternatives as well as the individuals, provides a continuous measure of the match probability as a function of uncertainty, and can be easily expanded to incorporate cases of interventions of the central planner, changes in the alternatives, and population changes.        
 
To do so, we define a space $\attrib$ of attributes. Each dimension in this space is a numerical representation of a single attribute in the relevant context (e.g. level of religiousness, level of parental involvement, geographic location of the school). The space is a $\nattribs$ dimensional box with boundaries, representing the range of each attribute.  

In this space we place $J$ alternatives, (such as the various schools) giving to each alternative a point $P_j$ in this  $\nattribs$-dimensional space.The location of an alternative represents its performance on each of the attributes. 
Each alternative has its attraction basin, and these partition the space of attributes $\attrib$ into a Voronoi tessellation (Obake and Suzuki 1997 \cite{OS}; De Leeuw 2005 \cite{dL}).

The construction divides the space of attributes into Voronoi cells, which are the basins of attraction: 
\[
D_j=\{x\in \attrib: \dist(x,P_j)\leq \dist(x,P_k), \quad  \forall k =1,\dots, J \}. 
\]
These are convex polyhedra, with disjoint interiors, whose union is all of $\attrib$, see Figure~\ref{fig:Voronoi}a.
 
Individuals (say, the students, or their parents) are represented as points in the attribute space $\attrib$. The location of an individual $i$ in this space, denoted by $x$, represents the true desire, or the "ideal" product of the individual (that is, a hypothetical alternative which should maximize individual $i$'s utility). 
It reflects both the desired level of attributes, as well as the importance of the attribute to the individual.
We want to match individuals to the alternative  which most closely matches their preferences. A closest match would be an alternative  $P_j$ so that the distance between the individual's location $x\in\attrib$ is not greater than the distance to any other product, i.e., that resides within the same Voronoi cell. In utility terms, one can say that the utility derived from each actual alternative $j$ can be represented as a function of the proximity of individual $i$ to the location of alternative $j$. 
 
We assume that the location of the alternatives in space is known to the individuals and is also known to the central planner. This is a reasonable assumption since consumers these days have wide access, through social media, customer reviews, and other online resources to the specifications of the alternatives in their choice set (Bronnenberg, Kim and Mela 2016\cite{BKM}).

\subsection*{Modeling uncertainty}
We add uncertainty to this representation: individuals, being sure they know what they want, locate themselves in a perceived place, which is distant from their true location in the attribute space up to an uncertainty factor $\rho$. 
 
A common distinction is made in literature between uncertainty – which assumes the probability of each alternative is known, and ambiguity – where the individual also needs to assess the probability distribution from which the alternatives are drawn (Kahn and Sarin 1988) \cite{KS}. In this paper we do not deal with ambiguity. We assume that all alternatives are available and their properties are known. We describe uncertainty in the desired level of attributes, that is, uncertainty in preferences, and not in product location, product performance, or influence of uncontrolled factors.

Our setup has appeared in Computer Science, in the "nearest-neighbor search problem", which returns the nearest neighbor of a query point $x$ in a set of points $\mathcal P$ in $\R^d$. 
Both the data (the set of points $\mathcal P$) and the query (the point $x$) may be uncertain. 
For instance (see Beskales et al. 2008 \cite{BSI}),  in location-based services, a user may request the locations of the nearest gas stations. To protect the user's privacy, an area that encloses the user's actual location may be used as the query object, while gas stations (the data objects $\mathcal P$) have deterministic locations.  In contrast to our goals, this literature focused on algorithmic and complexity aspects of the problem, see for instance the recent paper by Agrawal (2006) \cite{Agrawal} and the references there.

Due to the uncertainty $\rho$, an individual with a true location $x$, has a perceived location at a point around $x$. The perceived location is a point randomly drawn from a uniformly distributed  ball of radius $\rho$ around the true location $x$. Note that neither the individual nor the central planner know the true location $x$. All they know is the perceived location. Even if individuals are aware of the uncertainty $\rho$, they can not reconstruct the drawing process. The uniformity assumption is required for the convenience of the formal analysis, and makes sense for a finite space and for the general case, where we assume zero information on the preferences. 
 
Thus, rather than a point in the space of attributes, we actually have a ball $B(x,\rho)$ of all points at distance at most $\rho$ which define a possible perceived location of the individual whose true location is $x$. By taking the shape of a ball, we assume that the uncertainty is equal in all dimensions. This is a reasonable assumption for a general space, with no specific information on the dimension.
However, even if the uncertainty is not equal in all dimensions, the uncertainty ball can be regarded as the circumscribed ball where $\rho$ is the uncertainty in the dimension with the maximal uncertainty. 

Note, that while the uncertainty ball is uniform across attributes, and has a single radius for the entire population, the random draw of the perceived location generates heterogeneity across individuals: the perceived location is drawn for each and every individual separately, and therefore the \textit{ actual } error, namely, the distance between the perceived location and the true location varies across individuals and across attributes. The uncertainty $\rho$ can therefore be regarded as the maximum possible error in the perception.    
 
The point $x$ usually lies in a unique Voronoi cell $D_j$  which gives the correct match, while the ball $B(x,\rho)$ may intersect with some other cells.The probability $P_\rho(x)$ that the individual whose true location is the point $x$ is assigned to the correct Voronoi cell to which it belongs (that is, the cases where the choice of the individual is indeed  optimal) is the relative area (or volume) of the ball which lies in that cell, see Figure~\ref{fig:Voronoi}b: 
\begin{equation}\label{def of Prho rev}
    P_\rho(x) = \frac{\Vol(D_j\cap B(x,\rho))}{\Vol B(x,\rho)} .
\end{equation}

As illustrated in Figure~\ref{fig:Voronoi}, $P_\rho(x)$ strongly depends on the distribution of products in the attribute space, on the distance from the cell boundaries, and on the relationships between $\rho$ and the   location within the Voronoi cell. 

\subsection*{The probability for correct match}
From the perspective of the central planner which provides and supervises the services, a key measure of interest would be the effect of the uncertainty in individuals' preferences on the overall mismatch for the entire population.
A key measure we calculate is 
the {\em average} probability of correct match $\ave{P_\rho}$: 
\begin{equation}\label{def of Prho}
\ave{P_\rho} := \frac 1{\Vol(\attrib)}\int_\attrib P_\rho(x)dx
\end{equation}
that is the average of $P_\rho(x)$ over the entire attribute space - all the locations $x$, and all the Voronoi cells $D_j$.   
We seek to describe its variation as we change the uncertainty factor $\rho$. For a uniformly distributed population in the attribute space $\ave{P_\rho}$ is given by: 

\begin{equation}\label{eq:average prob}
    \ave{P_\rho} = \frac 1{\Vol(\attrib)}\sum_j \int_{D_j} \frac{\mathrm{vol}\left( B(x,\rho)\cap D_j\right)}{\mathrm{vol}B(x,\rho)}dx
\end{equation}
where $B(x, \rho)$ is the ball around $x$ of radius $\rho$,and $\int_{D_j} dx$ means integration within a Voronoi cell $j$ (see Supplementary Information Part 1 Proposition 1 for details).

To provide an intuition as to how to compute this integral, recall that for a given individual in location $x$, when $x$ is distanced more than  $\rho$ to the boundary, a match will always be obtained. However, when $x$ is closer to the boundary than $\rho$, the probability for a  mismatch grows. As illustrated in Figure~\ref{fig:match_prob}a, for each cell, there is only a finite “danger zone”, around its boundaries, where  a mismatch can occur. The cumulative area of the danger zones of all cells depends on two factors: 1) the size of $\rho$, 2) the total length of the boundaries between cells (in a general $K$ dimensional space the danger zone will be the relevant volume, and the length will be in dimension $K-1$. In the one-dimensional case, where we only have one attribute, and the boundary consists of isolated points, the "length" of the boundary will be the number of points).
For example, in Figure~\ref{fig:match_prob}a, describing a two dimensional space, this factor will be the total length of all the internal boundary segments between cells. When $\rho=0$, clearly $P_0\equiv 1$ as there is no uncertainty. When $\rho\gg 0$ is sufficiently large so that the uncertainty ball exceeds the combined size of the cells, the true location could be practically in any of the cells, meaning that the uncertainty is so vast that for every individual all the options seem reasonable to choose from.  

Therefore, if $\rho$ is small enough to disregard overlap of danger zones from different cells, the volume of the total danger zone is approximately (to leading order) given by $\rho$ times the total area of the internal boundaries ($\partial^{\rm int} D$).
For the special case of a single attribute (one dimensional space), the attribute space is an interval of the size length($\attrib$), and the Voronoi cells are segments within this interval. The boundaries are single points, so the total area of the boundaries is given directly by the number of alternatives $J$. We give an analytic formula for  $\ave{P_\rho}$ for the case of a small $\rho$ (namely, $\rho$ smaller or equal to half of the smallest segment  (See Supplementary Information Part 1 Proposition 2): 
  
 \begin{equation}\label{One dim small rho}
\ave{P_\rho}  =  1-\frac{J-1}2 \frac{\rho}{\length \attrib}.
 \end{equation}

For higher dimensions $\nattribs\geq 2$, we compute the match probability for the {\em first variation} of $\ave{P_\rho}$, that is for the slope at $\rho=0$, which is the  $v$ in the expansion $\ave{P_\rho}\sim 1- v\rho$. 
The notation $f(\rho)\sim g(\rho)$ as $\rho \to 0$  means $\lim_{\rho\to 0} f(\rho)/g(\rho)=1$.

In dimension $\nattribs \geq 2$, the mean probability for correct assignment $\ave{P_\rho}$ for $\rho$ small is
\begin{equation}\label{formula for small rho}
\ave{P_\rho} \sim 1- \left( \frac{c_{\nattribs}}{\Vol \attrib} \sum_j \Vol_{\nattribs-1}(\partial^{\rm int} D_j ) \right)\cdot \rho  ,\quad \rho \searrow 0  
\end{equation} 

where 
\begin{equation}\label{formula for c_K1}
c_{\nattribs}  = \frac 12 \frac{\Gamma \left(\frac{\nattribs}{2}+1\right)}{\sqrt{\pi } \Gamma \left(\frac{\nattribs+3}{2}\right)}  =
\begin{cases}
\frac 1{\pi} \frac{2^{2m}}{(m+1)\binom{2m+1}{m}},& \nattribs=2m\;{\rm even} \\
\\
  \frac{1}{2^{2m+2}}\binom{2m+1}{m}, &\nattribs=2m+1\;{\rm odd}. 
\end{cases}
\end{equation}

Here, $\Gamma$ is the Gamma function, thus $c_1=\frac1{4}$, $c_2=\frac2{3\pi}$, $c_3=\frac3{16}$, etc. See  Supplementary Information Part 1 Proposition 3 for the proof. When $\rho$ is large, we can no longer disregard the overlap of the different danger zones, and we rely on numerical calculation of equation \eqref{eq:average prob}.

Note that in the special case of ($\nattribs=1$), Eq~\eqref{formula for small rho} reduces to Eq~\eqref{One dim small rho} as $\ave{P_\rho}=1- (\frac 14\sum_j\#\partial^{\rm int} D_j ) \cdot \rho$, once we note that $c_K=\frac 14$. The boundary of an interior interval consists of $2$ points, so that $\#\partial^{\rm int} D_j=2$ for the $J-2$ interior intervals, and $\#\partial^{\rm int} D_j=1$ for the two intervals at the boundary of the space -  $j=1$, and $j=J$. 

Eq~\eqref{formula for small rho} reveals the dominance of the cell boundaries on the match probability. It predicts that as $\rho$ increases, $\ave{P_\rho}$ decreases linearly, with a slope that depends strongly on the length of the boundaries between the different Voronoi cells.

To extend the above analysis for the all values of $\rho$ we numerically calculate $\ave{P_\rho}$ for the two dimensional case. We represent the market as a two dimensional grid, with 6 alternatives located as shown in Figure \ref{fig:Voronoi} (the results are robust across location choices). We then execute three steps: first we assign for each grid point the best matched alternative. Second, we evaluate $P_\rho(x)$ by measuring the percentage of points having the same alternative in a sphere of radius $\rho$. Finally, we average $P_\rho(x)$ over the entire grid to obtain $\ave{P_\rho}$. 

Figure~\ref{fig:match_prob}a shows $P_\rho(x)$ for the setting described in Figure~\ref{fig:Voronoi}, for $\rho$ = 0.075. While most of the attribute space enjoys a perfect probability for a match, near the boundaries the probability decreases. Panel b describes $\ave{P_\rho}$ as a function of $\rho$ for the same market configuration, comparing the small $\rho$ approximation to numerical calculations. The slope of $\ave{P_\rho}$ vs. $\rho$ that is obtained from the approximation matches precisely the result of the  numerical simulation. Both analytical and numerical calculations show that the probability for a correct match rapidly decreases with $\rho$. While for the approximation, the decrease is linear, the numerical simulations show that for large values of $\rho$, the decrease is attenuated, saturating at $\sum_{j=1}^J (\Vol D_j)^2$.  

To illustrate the implications of the mismatch think of the opening example of choosing a school. In this setting, With $\rho=0.15$, ~20\% of the population will be dissatisfied, on average, with their choice, while there is another available school which matches their needs.

Note, that the matching in the above analysis is binary, that is, a mismatch happens when not assigning an individual with the true alternative, regardless of how the assigned alternative is close to the individual in the attribute space (this is an assumption in some of the literature on post-purchase evaluation e.g. Inman, Dyer and Jia (1997) \cite{IDJ}). In the Supplementary Information Part 2, we explore our results when the metric for the evaluation of the effect of uncertainty considers also the distance to the various alternatives.

\subsubsection*{Dependence on the number and distribution of alternatives}
The results shown in Figure \ref{fig:match_prob} provide an example for a specific configuration of six products. To assess the generalizability of this example we examined the effect of the number and distribution of the alternatives on the match probability. The slope of $\ave{P_\rho}$ where $\rho = 0$ serves as a useful metric, since it can be calculated directly from the length of boundaries. Higher slope indicates a stronger effect of the uncertainty on the match probability. Increasing the number of alternatives increases the slope - when more alternatives are available, the probability for a correct match decreases (see Figure~\ref{fig:sigma_and_N}a). This might seem counter-intuitive, as one would expect that more alternatives to choose from imply greater overall possibilities for a match. However, at the same time, more options mean more probability for a mismatch - as an individual is surrounded by more alternatives, he is less likely to choose the optimal one . In our terminology, we say that the ball of uncertainty intercepts with a larger number of Voronoi cells. Note, that there is a body of literature on the relationship between the number of alternatives during choice process, and the level of satisfaction and regret. Having more choice alternatives to choose from often increases the difficulty of the task and reduces satisfaction (e.g. Schwartz 2003 \cite{Schwartz}; Haynes 2009 \cite{Hayes}).

We further use the numerical simulations to explore how the distribution of the alternatives in the attribute space affects the match probability. Assume the location of the alternatives is drawn from a trimmed Gaussian distribution with width $\sigma$. Figure~\ref{fig:sigma_and_N}b presents the slope  $-d\ave{P_\rho} / d\rho\Big|_{\rho=0}$ vs. $\sigma$ for a market with 6 alternatives.  Increasing the width of the distribution increases the slope, thus reducing the probability for a correct match. The limiting case of uniform distribution has the lowest probability for a match (see Figure~\ref{fig:sigma_and_N}b). The intuition behind this is that the more dense the alternatives are, they are more similar to each other, meaning that the effective number of real alternatives is small, which, as illustrated in panel a, implies a higher match probability.  

\subsection*{Allocating service representatives} 
The results described above indicate that under uncertainty in preferences, mismatches are very likely to occur and can affect a considerable portion of the population, which creates a challenge for the central planner. As explained above, the authorities often employ service representatives (reps, hereafter), which assist individuals in understanding their true needs through personal meetings. Thus, the central planner wishes to improve $\ave{P_\rho} $ by introducing meetings with reps, which once having met with an individual, improve the individual's uncertainty from $\rho$ to a lower value  $\rho_l<\rho$.
Same as with the original uncertainty ball, our formulation practically allows heterogeneity in the amount of improvement:  after the meeting with the service rep, a new perceived location is drawn, within a smaller radius $\rho_l$. The actual amount of improvement will naturally vary for each individual and each dimension.

Due to budget constraints these reps meet only a fraction $b$ of the total population of individuals. We therefore ask who are the individuals which, within a given budget, should receive assistance from a rep in a way that will maximize the number of individuals who find their best matching alternative.

When the reps are allocated randomly, the new expected probability of correct assignment is   
\[ 
(1-b)\ave{P_{\rho}}+b\ave{P_{\rho_l}}
\]
Therefore, if we fix $\rho_l$ and $\rho$, and assuming that reps are randomly assigned to the population, increasing the proportion $b$ of reps results in a {\em linear} increase of  the expected probability of correct assignment.

We now check whether the central planner can improve the effectiveness of the reps by assigning them to specific individuals. To find the optimal assignment of service reps we define the local increase in match probability obtained from assigning a service rep to location $x$ to be  $\Delta(x,\rho,\rho_l) = P_{\rho_l}(x)- P_{\rho}(x)$.  Next, we choose $bN$ grid points, where $N$ is the total number of points on the grid, that have the maximal value of $\Delta(x,\rho,\rho_l)$, and reduce the uncertainty at these points to be $\rho_l$. Finally, to calculate the improvement in the match probability obtained from this process, we average $P_\rho(x)$ over the entire grid. We note that this optimal allocation scheme uses the true location $x$ of each individual, since we want to find the optimal allocation and spot the individuals who will have the maximum benefit from the service reps. In practice, as we stated above, $x$ is not known to the central planner, and thus, the central planner's implementation will be approximate, having its own error. We do not deal with such implementation error, but rather find the allocation which sets an upper limit to the benefit of the use of service reps.

Figure \ref{fig:service_reps} describes the overall improvement in $\ave{P_\rho}$ for various budget values $b$, where a budget is measured as the overall proportion of available rep meetings for the entire population. Panel a illustrates the areas which found to be optimal for receiving a meeting with the rep, within a budget $b=0.2$, for $\rho_l=0.05$ and $\rho=0.3$. We see that the places for optimal allocation (in blue), are those that are close to the boundaries between the Voronoi cells (white), but are not directly on the boundaries. When the distance from the boundary is smaller than $\rho_l$, meeting a rep will not significantly increase $\ave{P_\rho}$.
Panel b presents  $\ave{P_\rho}$ as a function of the budget $b$. While with random allocation, the improvement is linear with the budget, with the optimal allocation the curve shows a diminishing return and saturation at $b\approx 0.7$, meaning that the gain from allocating a service rep decreases as the number of allocated reps increases. 

To further demonstrate the effectiveness of service reps, we compare two ways to increase the match probability: the first is allocation reps as discussed, and the second is reducing the overall uncertainty of the population through means such as educational or citizen involvement programs.  Panel c shows, for each budget, what is the uncertainty $\rho$ that is equivalent to $b$ percentage of the population meetings with reps. 
A budget that allows meeting reps for 20\% of the population increases   $\ave{P_\rho}$ from 0.8 to 0.88, which is equivalent to reducing $\rho$ for the entire population from 0.3 to 0.18. While in practice such a change in the entire population might require long term educational and citizen involvement programs, the same result could be obtained by providing a relatively simple, easy to operate, front-desk service to a pre-targeted population. 

\section*{Discussion}
This paper deals with the overall impact of decisions, when individuals choose between alternatives, but have uncertainty as to the level of attributes that match their preferences.  

We add to previous literature by suggesting a continuous measure for the probability of a correct match, in a modeling framework that considers the entire set of alternatives, attributes, and individuals, and can help central planners in designing their policies.  We describe the attribute space as a Voronoi tessellation and use rigorous analysis and numerical simulations to describe the probability for correct match in space as a function of the uncertainty, and to calculate the average percentage of matches. We find that the overall mismatch can be considerable even for low levels of uncertainty, and thus can be a concern for policy makers. We further explore a practice often used by central planner - allocating service representatives to help individuals obtain the correct decisions. We use numerical simulations to show that within a given budget, the allocation is most effective for individuals whose preferences are at a certain distance from the boundaries of a Voronoi cell - not too deep in the cell, but yet not too close to a boundary. 
 
This paper suggests several avenues for future research. First, one could re-examine our assumption on a uniform distribution of the population in the attribute space. Other distributions, such as bell-shaped distribution around a central value might diminish the impact of uncertainty (if, for example, there are several clusters of individuals and a single alternative is placed in the middle of each cluster), or alternatively enhance it (if preferences are centered around certain values, but the alternatives are scattered in space). An additional extension could be exploring the issue of capacity constraints - the scenario in which a mismatch could prevent {\bf another} individual from being correctly matched. A third topic of interest would be endogenous sources of information, beside the reps, such as word-of-mouth from other users. Since this additional information also has uncertainty, it can hypothetically work in both directions and its influence on the reps allocation is not trivial.

\section*{Acknowledgements}
T.D. is grateful to the Azrieli Foundation for Azrieli Fellowships and is supported by a quantum science and technologies fellowship given by the Israeli council for higher education.
R.P was supported by the Israeli Science Foundation and by the KMart foundation of the Hebrew University. 
Z.R. was supported  by an Advanced Grant from the European Research Council under the European Union's Horizon 2020 research and innovation programme/ERC grant agreement n$^{\text{o}}$~786758.
The authors thank Eliya Horn for her research assistance. 

\section*{Author contributions statement}

T.D., R.P., and Z.R. worked jointly and contributed equally to the paper. All authors reviewed the manuscript. 

\section*{Additional information}
\textbf{Competing interests} The authors declare no competing interests.

\begin{figure}[ht]
\begin{center}
\includegraphics[width = \textwidth]{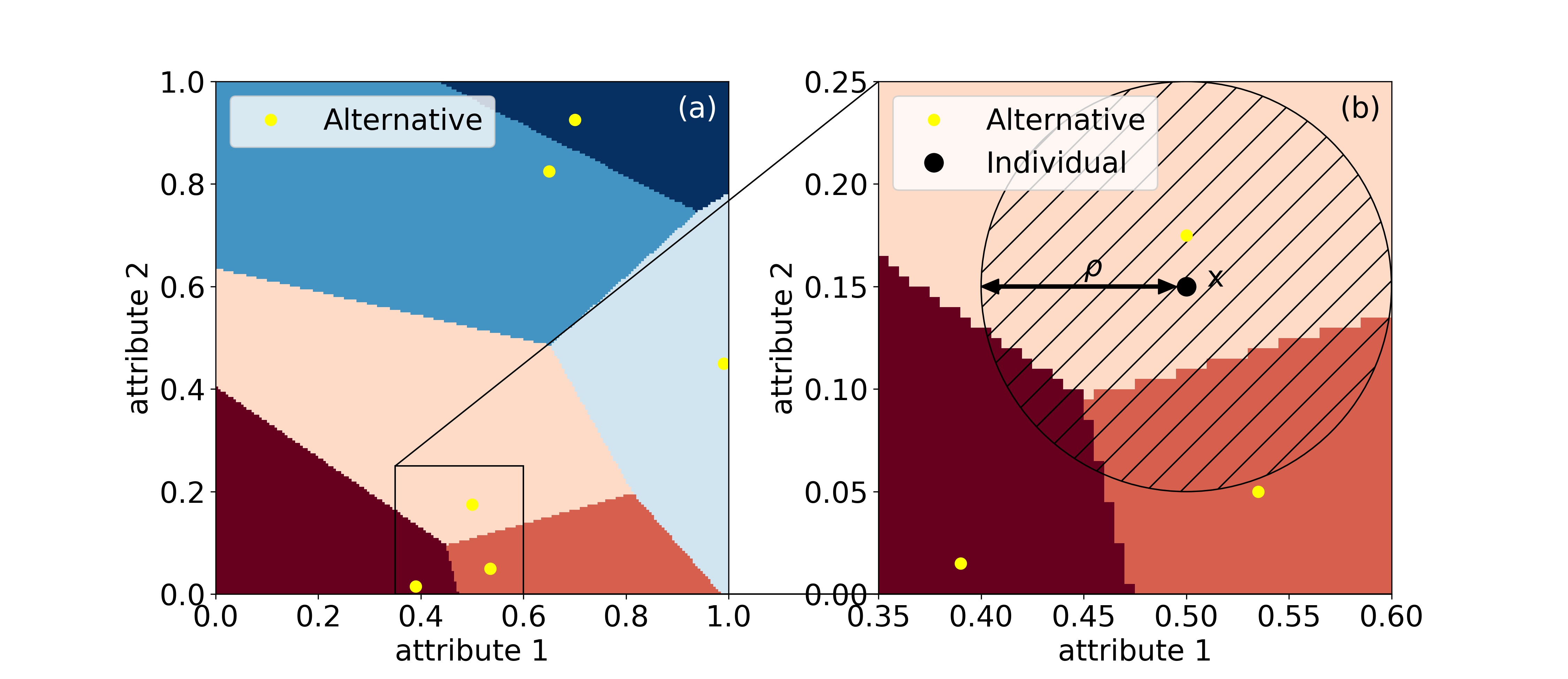}
\caption{\textbf{Voronoi tessellation. a.} An example for a two dimensional square $[0,1]^2$ of side length $1$, where 6 alternatives (yellow) divide the area to distinct Voronoi cells. \textbf{b.} In this example, $\rho = 0.1$. The probability $P_\rho(x)$ that an individual $x$ (marked by the black dot) chose the correct Voronoi cell is the relative area of the part of the ball of radius $\rho$ around $x$  which lies in the same Voronoi cell as $x$. }
\label{fig:Voronoi}
\end{center}
\end{figure}

\begin{figure}[ht]
\begin{center}
\includegraphics[width = \textwidth]{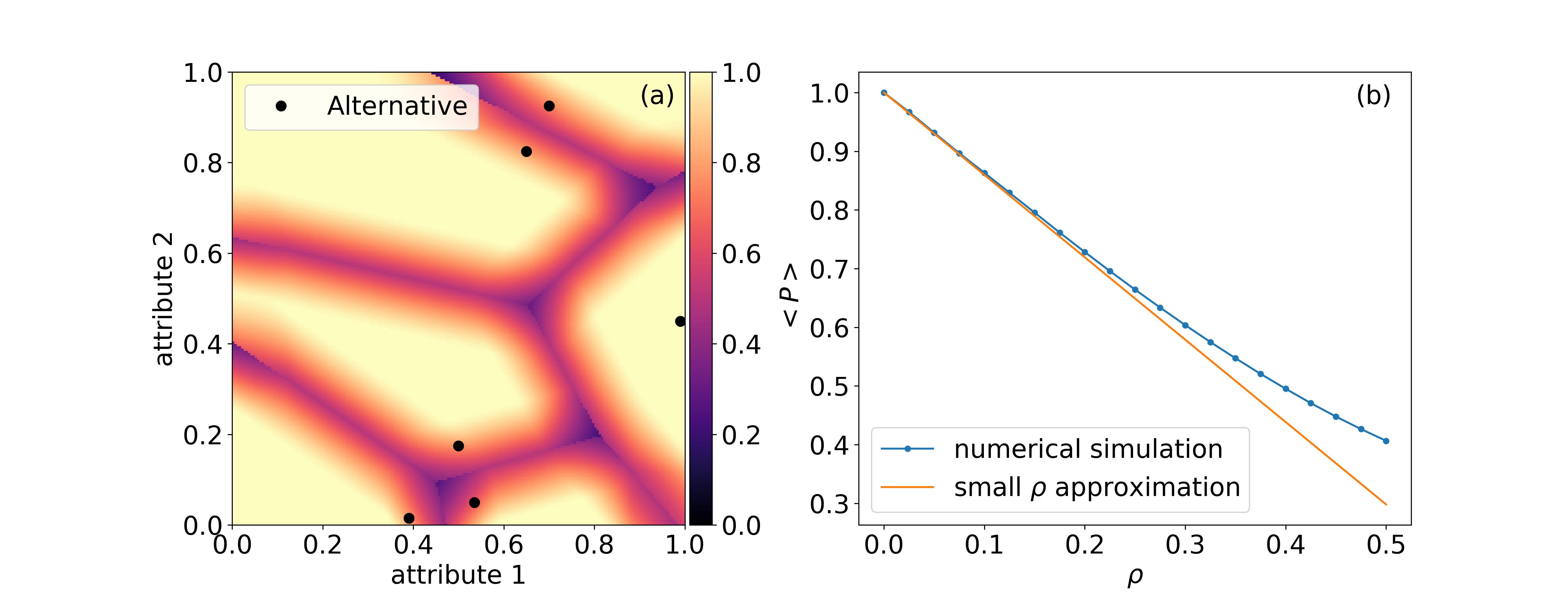}
\caption{\textbf{Match probabilities. a.} The local probability for a match, $P_\rho (x)$, is plotted as a color map for the example shown in Figure.~\ref{fig:Voronoi}.  \textbf{b.}   Average  probability for a match $\ave{P_\rho}$ as a function of $\rho$ for this configuration. Displayed is a comparison between the small $\rho$ linear approximation and the numerical calculation.  The probability for correct match rapidly decreases with $\rho$ and the decrease is attenuated for large values of $\rho$, until saturation.
}
\label{fig:match_prob}
\end{center}
\end{figure}

\begin{figure}[ht]
\begin{center}
\includegraphics[width = \textwidth]{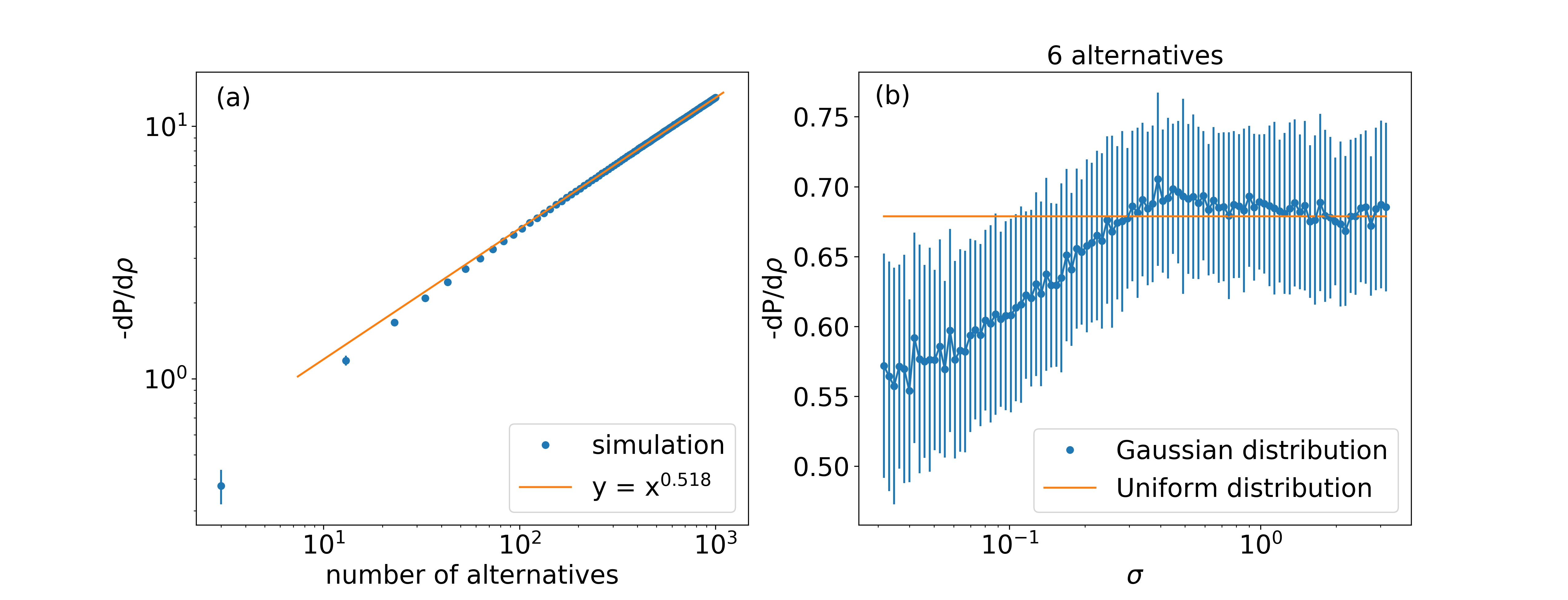}
\caption{\textbf{Effects of the distribution and the number of alternatives on the probability for a match a.} $ -d\ave{P(\rho=0)} / d\rho$ vs. the number of alternatives. For each number of alternatives we generated 100 market configurations sampled from a uniform distribution. For each configuration we calculated the length of the boundaries between the resulting Voronoi cells and used equation \ref{formula for small rho} to compute $d\ave{P(\rho=0)} / d\rho$. We present the average value of the different configurations. The error bar shows the standard deviation. \textbf{b.} Dependence on the distribution of alternatives: $ -d\ave{P(\rho=0)} / d\rho$ vs. $\sigma$, where $\sigma$ is the width of a trimmed Gaussian distribution, from which the location of alternatives is sampled. The simulation procedure is similar to panel (a). We vary $\sigma$ for the case of 6 alternatives (panel b)}
\label{fig:sigma_and_N}
\end{center}
\end{figure}

\begin{figure}[ht]
\begin{center}
\includegraphics[width = \textwidth]{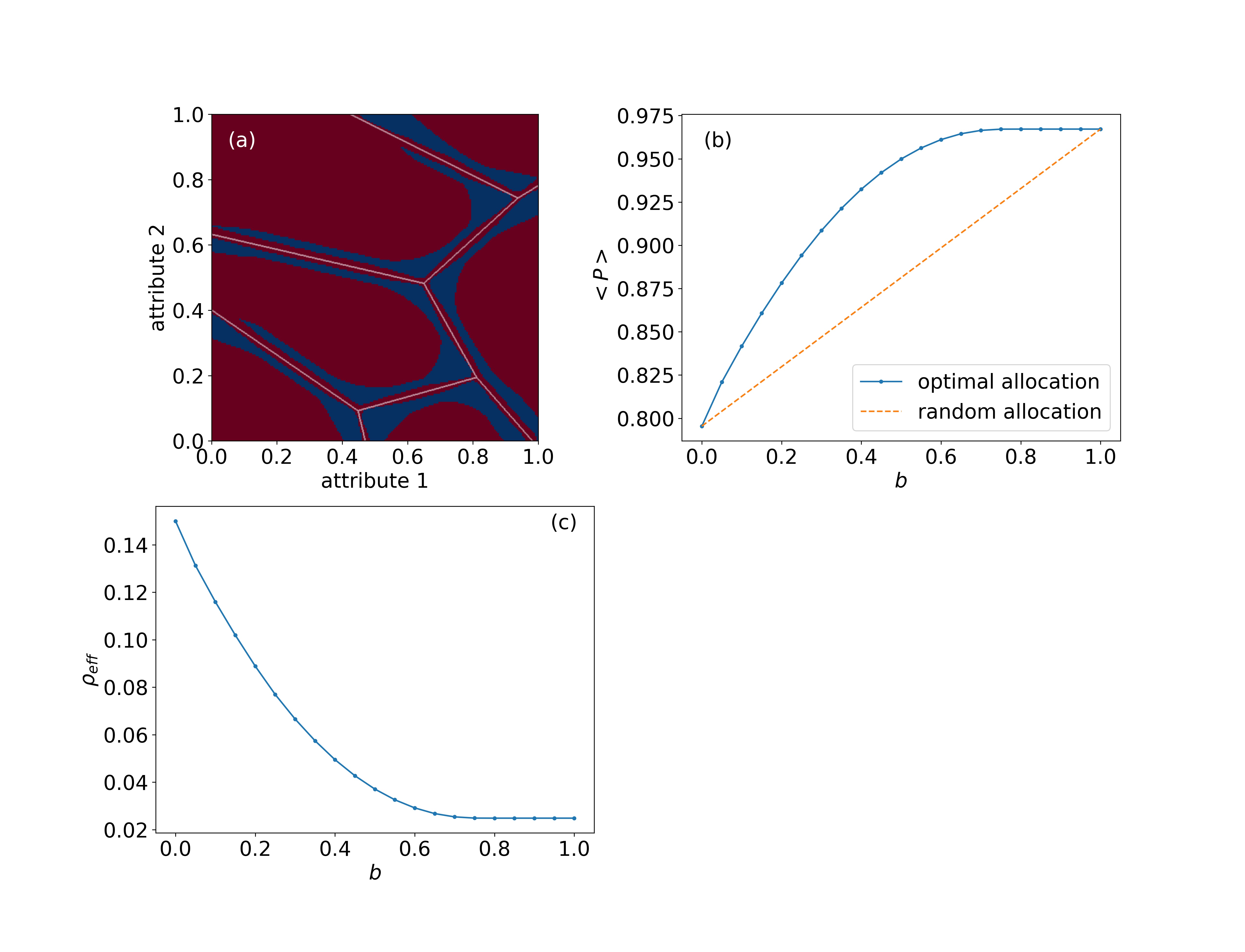}
\caption{\textbf{Improving match probabilities using service reps a.} Areas which maximize the effectiveness of service reps (blue), within a budget $b=0.2$, for $\rho_l=0.025$ and $\rho=0.15$   \textbf{b.}   Average  probability for a match $\ave{P_\rho}$ as a function of the budget $b$, where $\rho_l=0.025$ and $\rho=0.15$. The orange line shows a linear improvement when the reps are assigned randomly. The blue dots show maximal improvement when the reps are allocated optimally. \textbf{c.} The equivalent overall $\rho$ which results in the same $\ave{P_\rho}$ as an optimal allocation of reps within a given budget $b$.  }
\label{fig:service_reps}
\end{center}
\end{figure}

\newpage
 
\appendix

\section{Proof of the formula for the first variation of the expected probability of correct assignment}

% We give a proof of the formula for the first variation of the expected probability of correct assignment. 

We start with a space of attributes $\attrib$ which is a $\nattribs$-dimensional box, $\attrib=[a_1,b_1]\times[a_2,b_2]\times[a_\nattribs, b_\nattribs]$ with side-lengths $b_i-a_i$. We are given partition of the space of attributes $\attrib$   into   Voronoi cells with disjoint interiors: $\attrib= \coprod_{j=1}^J D_j$. The cells are convex polytopes, so that the boundary of each cell is covered by finitely many hyperplanes. %In the application at hand, this is the Voronoi tessellation into basins of attraction. 

We fix an uncertainty factor, a ball of radius  $\rho>0$, and for a point $x\in \attrib$ we ask what is the probability $P_\rho(x)$ that we assign  the correct  Voronoi cell,  given this uncertainty factor $\rho$? That is, given that $x\in D_j$, what is the probability  
that we assign $D_j$ as the basin of attraction, using an error bar of $\rho$?  Note that the problem only makes sense for small $\rho$, because once $\rho$ is sufficiently large so that the ball $B(x,\rho)$ covers all of $\attrib$, say $\rho>\rho_{\max}$, then the question is independent of $\rho$.

 We can write a formula for the expected value $\ave{P_\rho} $ of $P_\rho(x)$ (as we average over $x$): 
\begin{prop}\label{prop:gen formula}
%For $\rho<\rho_{\max}$,  
\begin{equation}\label{gen formula}
\ave{P_\rho} = \frac 1{\Vol \attrib} \sum_j  \int_{D_j} \frac{\Vol\Big(D_j\cap B(x,\rho)\Big)}{\Vol B(x,\rho)}dx 
\end{equation}
where $B(x, \rho)$ is the ball around $x$ of radius $\rho$.  
%and $\chi_\rho(x, y)=1$ if $\dist(x,y)\leq \rho$, and $\chi_\rho(x,y)=0$ otherwise.  

If $\rho>\rho_{\max}$, then $\ave{P_\rho} $ saturates at 
$\ave{P_\rho} = \sum_{j=1}^J ( \Vol D_j )^2/ (\Vol  \attrib)^2$. 
\end{prop}
\begin{proof}
 To see \eqref{gen formula}, recall that  given that $x\in D_j$,   the probability $P_\rho(x)$ %$p(D_j,\rho)(x)$ 
that we select $D_j$ as the basin of attraction  is  the relative area of the intersection of the ball of radius $\rho$ with the cell $D_j$: 
\[
P_\rho(x)
%p(D_j,\rho)(x)=
% \Prob(x\in B(x,\rho)\cap D_j) \Big| x\in D_j) = 
=\begin{cases}\frac{\Vol\left(D_j\cap B(x,\rho)\right)}{\Vol B(x,\rho)},&x\in D_j\\   \\
0,&x\notin D_j . \end{cases}
\]

We want to compute the expected value of $P_\rho(x)$ (we average over $x$):
\[
%\begin{split}
\ave{P_\rho} %&=\sum_j \E(p(D_j,\rho)(x)  ) 
%\\&\mbox{where }\E(\bullet |x\in D_j) \mbox{ denotes the conditional expectation}
%\\&
= \frac 1{\Vol \attrib} \sum_j  \int_{D_j}P_\rho(x) dx
 =\sum_j \frac 1{\Vol \attrib}  \int_{D_j} \frac{\Vol\Big(D_j\cap B(x,\rho)\Big)}{\Vol B(x,\rho)}dx .
%\end{split}
\]

To see the saturation value, take $\rho$ to be larger than the diameter of the space of attributes $\attrib$, so that for each point $x$, the ball $B(x,\rho)$ coincides with all of  $\attrib$. %Applying Formula~\eqref{gen formula}
Then for each $x\in \attrib$, the intersection $B(x,\rho)\cap D_j=D_j$, and we can compute $\ave{P_\rho}$ simply as a conditional expectation, by writing
\[
P_\rho(x) = \sum_{j=1}^J \mathbf 1_{D_j}(x) \frac{\Vol D_j}{\Vol  \attrib}
\]
and then 
\[
%\begin{split}
\ave{P_\rho}  %&
= \frac {\int_{\attrib}P_\rho(x) dx }{\Vol   \attrib}
 =  \frac 1{\Vol   \attrib}\sum_{j=1}^J \int_{  \attrib }\mathbf 1_{D_j}(x) \frac{\Vol D_j}{\Vol   \attrib}dx 
%\\&  
= \sum_{j=1}^J \Big( \frac{\Vol D_j}{\Vol  \attrib}\Big)^2
% \end{split}
\]
as claimed. 
\end{proof}

\subsection{The one dimensional case}
Equation \eqref{gen formula} makes sense in any dimension,   
but it is only in dimension $\nattribs=1$, when the space of attributes  is an interval $\mathcal A=[0,L]$, that we know how to extract an exact expression from it for small $\rho$.   
  
 \begin{prop}\label{prop:One dim box}
 For $\nattribs=1$, and $\rho<\tfrac 12  \min_j(a_{j+1}-a_j)$, 
 \[
 \ave{P_\rho}  =  1-\frac{J-1}2 \frac{\rho}{\length \attrib} .
 \] 
 \end{prop}
\bigskip

\begin{proof}
 
We use equation~\eqref{gen formula}: %Take $\attrib = [0,L]$, $L=\length \attrib$. 
In this one-dimensional case, the Voronoi cells are intervals $D_j=[a_{j},a_{j+1}]$, with $0=a_1<a_2<\dots  <a_{J+1}=L$. We assume that 
\[
2\rho< \min_{j} {\rm length} D_j = \min_j(a_{j+1}-a_j) . 
\]

To compute the contribution of each cell $D_j=[a_j,a_{j+1}]$, divide the region of integration into an interior region $D_j^{\rm int}:=[a_j+\rho,a_{j+1}-\rho]$ and two boundary regions $[a_j,a_j+\rho]$ and $[a_{j+1}-\rho,a_{j+1}]$. 

For $x$ in  the interior region, we have $B(x,\rho)\subset D_j$ so that $D_j\cap B(x,\rho) = B(x,\rho)$ and hence 
\[
\begin{split}
\int_{D_j^{\rm int}} \frac{\length\Big(D_j\cap B(x,\rho)\Big)}{\length B(x,\rho)}dx  &= \int_{D_j^{\rm int}} 1 dx 
= \length (D_j^{\rm int}) \\
&= (a_{j+1}-\rho)-(a_j+\rho)
\\
&=a_{j+1}-a_j-2\rho = \length(D_j)-2\rho .
\end{split}
\]

To compute the contribution of the boundary components, note that there are two types, corresponding if they coincide with the boundary of the interval, namely $j=1,J+1$, or not ($j=2,\dots,J$). 
\begin{figure}[ht]
\begin{center}
\includegraphics[height=40mm]{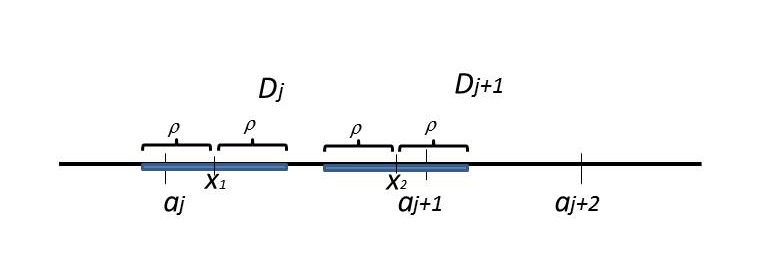}
\caption{ The overlap of the ``balls'' $B(x,\rho)=[x-\rho,x+\rho]$ with the Voronoi intervals $D_j=[a_j,a_{j+1}]$.}
\label{fig:intervals}
\end{center}
\end{figure}

For components which do not intersect the boundary, namely if $j\neq 1,J+1$ then for all $x \in [a_j,a_j +\rho] \cup [a_{j+1} - \rho , a_{j+1}]$, the ``ball" $B(x,\rho) = [x-\rho,x+\rho]$ has length $2\rho$, but 
\[
B(x,\rho)\cap D_j = [x-\rho,x+\rho]\cap [a_j,a_{j+1}] = \begin{cases} [a_j,x+ \rho], & a_j\leq x\leq a_j+\rho  \\ [x-\rho,a_{j+1}], & a_{j+1}-\rho\leq x\leq a_{j+1} \end{cases}
\]
(see Supplementary Figure~\ref{fig:intervals}) and 
%\marginpar{add figure?}
\[
\length B(x,\rho)\cap D_j  = \begin{cases}  x+\rho-a_j , & a_j\leq x\leq a_j+\rho 
\\ a_{j+1}+\rho-x , & a_{j+1}-\rho\leq x\leq a_{j+1}
\end{cases}
\]
so that 
\[
\int_{a_j}^{a_j+\rho}  \frac{\length\Big(D_j\cap B(x,\rho)\Big)}{\length B(x,\rho)}dx  = \int_{a_j}^{a_j+\rho} \frac{x+\rho-a_j}{2\rho}dx = \frac 34 \rho
\]
and
\[
\int_{a_{j+1}-\rho}^{a_{j+1}} \frac{\length\Big(D_j\cap B(x,\rho)\Big)}{\length B(x,\rho)}dx  =\int_{a_{j+1}-\rho}^{a_{j+1}} \frac{a_{j+1}+\rho-x}{2\rho}dx = \frac 34 \rho.
\]
Altogether, we obtain for $ j\neq 1,J $ that 
\[
\int_{D_j} p_\rho(D_j)(x)dx  =   \length(D_j)-2\rho  + \frac 34 \rho+\frac 34 \rho = \length D_j-\frac 12 \rho. 
\]

For components which do intersect the boundary, that is for $D_1 = [0,a_2]$ or $D_J=[a_J,L]$, we have 
\[
B(x,\rho)=B(x,\rho)\cap D_1  = [0,x+\rho],\quad x\in D_1,
\]
and 
\[
B(x,\rho)= B(x,\rho)\cap D_J=[x-\rho,L] , \quad x\in D_J
\]
so that 
\[
\frac{\length B(x,\rho)\cap D_j}{\length B(x,\rho)} = 1, \quad x\in D_1\cup D_J
\]
and we get  a contribution of 
\[
\int_0^\rho 1dx = \rho = \int_{L-\rho}^L 1 dx . 
\]
Therefore, 
\[
\int_{D_j} P_\rho(x)dx  =  \length D_j -2\rho + \frac 34 \rho+\rho = \length D_j-\frac 14 \rho, \quad j=1,J .
\]

Altogether we find
\[
\begin{split}
\ave{P_\rho} &=\frac 1{\length \attrib}  \sum_{j=1}^J \int_{D_j} P_\rho(x)dx  
\\
&= \frac 1{\length \attrib}  \Big( \length D_1 -\frac 14 \rho + \sum_{j=2}^{J-1} \left( \length D_j -\frac 12 \rho\right) + \length D_J -\frac 14 \rho   \Big) 
\\&= \frac 1{\length \attrib}   \Big( \sum_{j=1}^J \length D_j - \frac{J-1}{2}\rho \Big)   = 1- \frac{J-1}{2}\frac {\rho }{\length \attrib}
\end{split}
\]
as claimed.
\end{proof}

\subsection{ Higher dimensions $\nattribs\geq 2$}  
We now pass to the higher dimensional case $\nattribs\geq 2$. Our goal in this section is to obtain an exact formula for the {\em first variation} of $\ave{P_\rho}$, that is for the slope at $\rho=0$. 
 
% We denote by 
%\[ 
%\ave{f(z)}_{z\in B(\rho)}:=\frac 1{\Vol B(\rho)}\int_{B(\rho)} f(z)dz
%\]
%the average of a function on the ball $B(\rho)=B(0,\rho)$. Let 
%\begin{equation}
%c_{\nattribs} = \frac 12 \ave{|z_1|}_{z\in B(1)}
%\end{equation}
%(the average over the unit ball).  We shall show (see Lemma~\ref{lem:ave |z|}) that

For each Voronoi cell $D_j$, we denote by $\partial^{\rm int} D_j$ the part of the boundary of $D_j$ which does not lie on the boundary of the box (the space of attributes)  $\attrib$.

 \begin{prop}\label{prop:expand near 0} 
In dimension $\nattribs \geq 2$, the mean probability for correct assignment $\ave{P_\rho}$ for $\rho$ small is 
\[
\ave{P_\rho} \sim 1- \left( \frac{c_{\nattribs}}{\Vol \attrib} \sum_j \Vol_{\nattribs-1}(\partial^{\rm int} D_j ) \right)\cdot \rho  ,\quad \rho \searrow 0  
\]
where 
\begin{equation}\label{formula for c_K}
c_{\nattribs}  = \frac 12 \frac{\Gamma \left(\frac{\nattribs}{2}+1\right)}{\sqrt{\pi } \Gamma \left(\frac{\nattribs+3}{2}\right)}  =
\begin{cases}
\frac 1{\pi} \frac{2^{2m}}{(m+1)\binom{2m+1}{m}},& \nattribs=2m\;{\rm even} \\
\\
  \frac{1}{2^{2m+2}}\binom{2m+1}{m}, &\nattribs=2m+1\;{\rm odd}. 
\end{cases}
\end{equation}
%Thus $c_1=\frac 14$, $c_2=\frac {2}{3\pi}$, $c_3 = \frac{3}{16}$, etcetera. 
% \begin{center}
% \begin{table}[h] 
% \begin{tabular}{ | l|  c|c|c|c|c|c|c|c|c | r |}
%    \hline
%    $\nattribs$ & $1$ & $2$ & $3$ & $4$ & $5$ &  $6$  & $7$  & $8 $ & $9$ & $ 10$ \\  
%     \hline   
%    \large    $c_\nattribs$ &    \large $\frac 14$ &  \large $\frac {2}{3 \pi}$ &  \large  $\frac {3}{16}$ &  \large  $\frac{8}{15 \pi}$ &  \large$\frac{5}{32}$ &  \large$\frac{16}{35\pi}$ & 
 %    \large$\frac{35}{256}$ &  \large$\frac{128}{315 \pi}$ &  \large$\frac{63}{512}$ & \large $\frac{256}{693 \pi}$ \\
 %   \hline
 %   
 %  \end{tabular}
%    \bigskip
 %  \caption{$c_\nattribs= \frac 12 \frac{\Gamma \left(\frac{\nattribs}{2}+1\right)}{\sqrt{\pi } \Gamma \left(\frac{\nattribs+3}{2}\right)}$} 
%\end{table}
%\end{center}

\end{prop}
\begin{proof}
We start by using equation \eqref{gen formula}
\begin{equation}\label{eq:average probA2}
    \ave{P_\rho} = \frac 1{\Vol(\attrib)}\sum_j \int_{D_j} \frac{\mathrm{vol}\left( B(x,\rho)\cap D_j\right)}{\mathrm{vol}B(x,\rho)}dx \;.
\end{equation}
%where $B(x, \rho)$ is the ball around $x$ of radius $\rho$, $\int_{D_j} dx$ means integration within a Voronoi cell $j$.
 There are two types of points $x\in D_j$: Type I, those $x\in D_j$ so that the ball $B(x,\rho)$ is entirely contained in $D_j$, and type II are the rest (Supplementary Figure~\ref{fig:type1type2}). Note that if  $x$ is close to the boundary of $\attrib$: $\dist(x,\partial \attrib)<\rho$, but far from the interior boundary of the cell, that is $\dist(x,\partial^{\rm int}D_j)>\rho$, then $B(x,\rho)\subseteq D_j$ is entirely contained in the cell, even though it is only a truncated ball (Supplementary Figure~\ref{fig:truncatedball}). This means that these points are type I.  Thus type II points are precisely those $x\in D_j$ so that $\dist(x, \partial^{\rm int}D_j)<\rho$. 
 \begin{figure}[ht]
\begin{center}
\includegraphics[height=60mm]{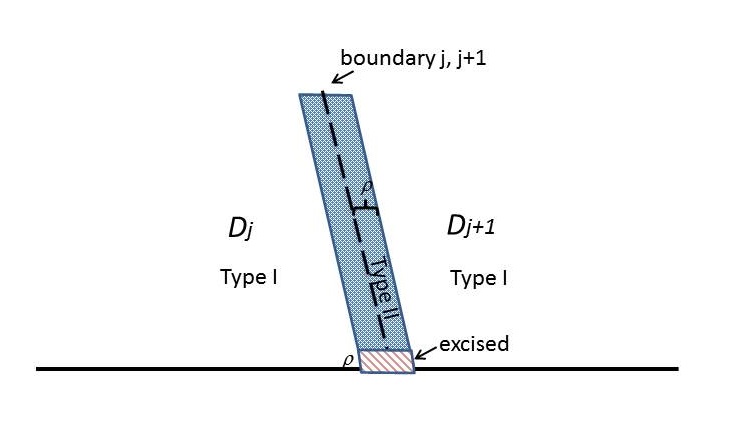}
\caption{Type I region, type II region (shaded) and the excised points near the boundary (shaded region with stripes). }
\label{fig:type1type2}
\end{center}
\end{figure}

 For type I points, we have $ B(x,\rho)\cap D_j =B(x,\rho)$ so that the quotient of volumes equals unity: 
 \[
 \frac{\mathrm{vol}\left( B(x,\rho)\cap D_j\right)}{\mathrm{vol}B(x,\rho)}=1, \quad x\;{\rm of\; type\; I}.
 \]
Thus the type I points contribute
\begin{multline}\label{contribution of type I}
\frac 1{\Vol(\attrib)}\sum_j \int_{\substack{x\in D_j\\x\;{\rm type\; I}}} \frac{\mathrm{vol}\left( B(x,\rho)\cap D_j\right)}{\mathrm{vol}B(x,\rho)}dx   = \sum_j \frac  { \mathrm{vol}\{ x\in D_j\;{\rm type\; I}\} }{\Vol(\attrib)} .
%\\
 %= \frac 1{\Vol(\attrib)}  \Big(\mathrm{vol}(D_j) - \mathrm{vol}\Big( x\in D_j\;{\rm type\; II}\Big)  \Big)
\end{multline}

The type II points are contained in a ``strip'' around the interior boundary of ``width'' $2\rho$. We excise the contribution of points  which are also $\rho$-close to $\partial \attrib$ or to more than one interior face (Supplementary Figure~\ref{fig:type1type2}). The volume of these points is bounded by $O(\rho^2)$, since they are at distance $\leq \rho$ from the intersection of two faces or the intersection of a face with $\partial \attrib$, which has codimension $2$. Since 
$\mathrm{vol}\left( B(x,\rho)\cap D_j \right)/\mathrm{vol}B(x,\rho)\leq 1$ in any case, the total contribution of such points is 
$O(\rho^2)$, which is negligible.  
%\marginpar{is this true in higher dimension?} which is negligible for $\nattribs \geq 2$. 
Thus we need only consider points $x$ with $\dist(x, \partial^{\rm int}D_j)<\rho$ and in addition that $B(x,\rho)$ is an actual Euclidean ball, not a truncated one.

\begin{prop}\label{prop type II} 
For $\rho$ sufficiently small, the contribution of type II points is
\[
 \sum_j \frac { \Vol(x\in D_j\;{\rm type\; II})} {\Vol(\attrib)}    -  \frac {c_\nattribs \Vol_{\nattribs-1} \partial^{\rm int} D_j}{\Vol(\attrib)}  \rho +O(\rho^2) .
\]
\end{prop}
 \noindent Putting together equation \eqref{contribution of type I} and Proposition~\ref{prop type II} gives Proposition~\ref{prop:expand near 0}. 
 \end{proof}
 
\subsection{Proof of Proposition~\ref{prop type II} }
Fix a component $H$ of the interior boundary $\partial^{\rm int} D_j$; $H$ is a hyperplane. After rotation, reflection and translation of the diagram, we may assume that the boundary component $H$ is the coordinate hyperplane $H=\{x=(x_1,\dots x_\nattribs):x_\nattribs=0\}$, and that the cell $D_j$ lies in the top half-space $H_+=\{(y_1,\dots, y_\nattribs): y_\nattribs\geq 0\}$ 
(Supplementary Figure~\ref{fig:truncatedball}).  Then for every $x \in H_+$, we have $\dist(x,H) = x_\nattribs$ and we assume that $0\leq x_\nattribs\leq \rho$. We need to compute 
\begin{multline*}
\int\limits_{x: (x_1,\dots,x_{\nattribs-1},0)\in H}  \frac{\mathrm{vol}\left( B(x,\rho)\cap D_j\right)}{\mathrm{vol}B(x,\rho)} dx_1\dots dx_\nattribs 
\\
= \frac 1{ \mathrm{vol}B(0,\rho)}  
\int\limits_{\substack{0\leq x_\nattribs\leq \rho\\  (x_1,\dots,x_{\nattribs-1},0)\in H}}
\mathrm{vol}\left( B(x,\rho)\cap D_j\right) dx_1\dots dx_\nattribs .
\end{multline*}

\begin{figure}[ht]
\begin{center}
\includegraphics[height=40mm]{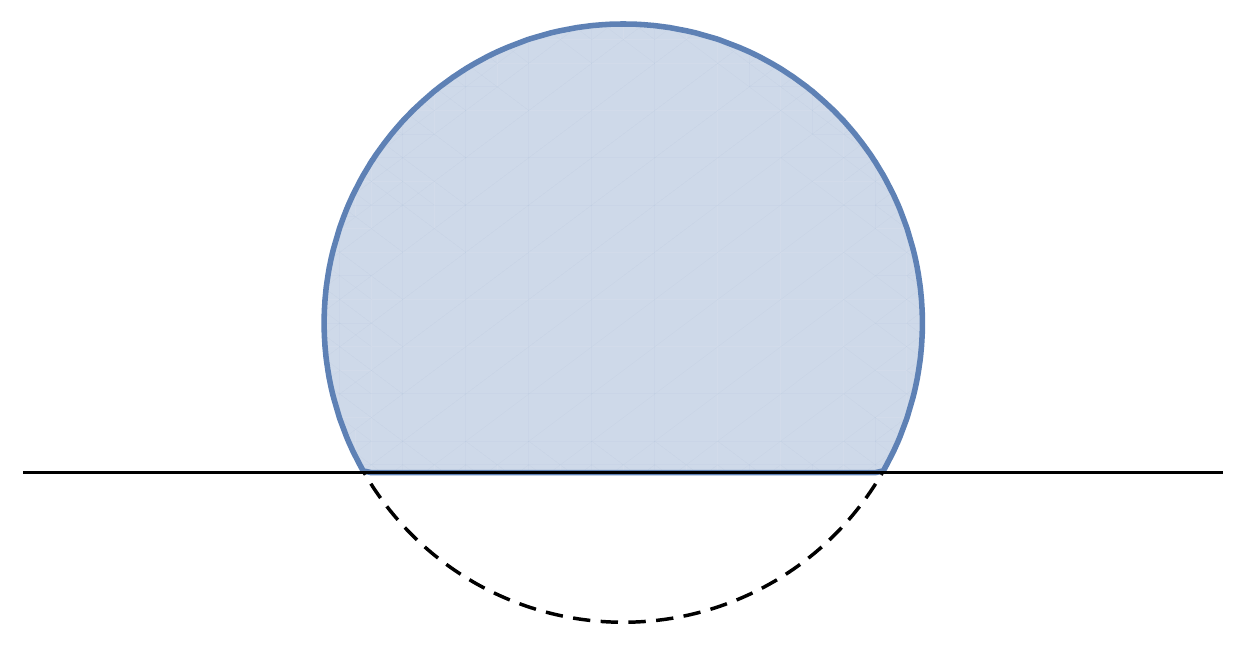}
\caption{ A truncated ball }
\label{fig:truncatedball}
\end{center}
\end{figure}

We fix the first $\nattribs-1$ components $x^0 = (x_1,\dots, x_{\nattribs-1})$, and compute  the integral over $x_\nattribs$: 
\begin{lem}
Fix $(x_1,\dots, x_{\nattribs-1})$ so that $(x_1,\dots,x_{\nattribs-1},0)\in H$. Then 
\begin{multline*}
 \frac{1}{\Vol B(0,\rho)}\int_{x_\nattribs=0}^\rho \mathrm{vol}\Big(  B\Big( \left( x_1,\dots, x_{\nattribs-1},x_\nattribs\right),\rho \Big)\cap H_+ \Big) dx_\nattribs  
 \\
 %= 1 -\frac{\Gamma(\frac \nattribs 2+1)}{2\sqrt{\pi}\Gamma(\frac{\nattribs+3}2)}\cdot \rho +O(\rho^2)
 =1-c_\nattribs  \cdot \rho +O(\rho^2) .
\end{multline*}
\end{lem}

\begin{proof}
Since the integral is independent of the first $\nattribs-1$ components, those may be taken to be zero, so that 
$(x_1,\dots, x_{\nattribs-1})=(0,\dots,0)$. So we want to compute
\[
 \int_{x_\nattribs=0}^\rho \mathrm{vol}\Big(  B\Big( \left( 0,\dots,0,x_\nattribs\right),\rho \Big)\cap H_+ \Big) dx_\nattribs . 
\]

The set $B\Big( \left( 0,\dots,0,x_\nattribs\right),\rho \Big)\cap H_+$ is the bigger half of the ball $B\Big( \left( 0,\dots,0,x_\nattribs\right),\rho \Big)$ (see Supplementary Figure~\ref{fig:truncatedball}); we find it easier to compute the integral over the complementary, smaller half, which is a spherical cap (Supplementary Figure~\ref{fig:smallcap}), and this in turns equals 
\begin{figure}[ht]
\begin{center}
\includegraphics[height=20mm]{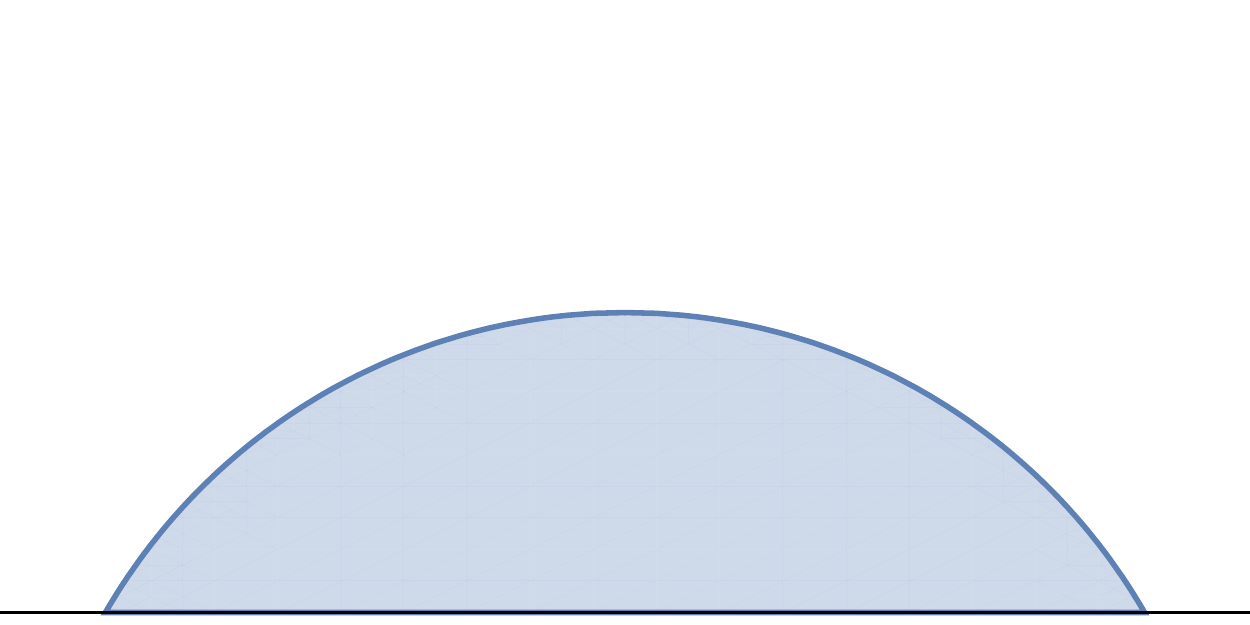}
\caption{ A spherical cap}
\label{fig:smallcap}
\end{center}
\end{figure}
\begin{multline*}
  \int_{x_\nattribs=0}^\rho \mathrm{vol}\Big(  B\Big( \left( 0,\dots,0,x_\nattribs\right),\rho \Big)\cap H_+ \Big) 
  dx_\nattribs  
  \\
  = 
  \mathrm{vol}B(0,\rho)- \int_{x_\nattribs=0}^\rho \mathrm{vol}\Big(  B\Big( \left( 0,\dots,0,-x_\nattribs\right),\rho \Big)\cap H_+ \Big) dx_\nattribs  .
\end{multline*}
Dividing by $ \mathrm{vol}B(0,\rho)$ gives
\[
1-\frac 1{ \mathrm{vol}B(0,\rho)} \int_{ 0}^\rho A(x_\nattribs) dx_\nattribs
\]
where $A(x_\nattribs)$ is the volume of the small spherical cap
\[
\begin{split}
 A(x_\nattribs)& := \mathrm{vol}\Big(  B\Big( \left( 0,\dots,0,-x_\nattribs\right),\rho \Big) \cap H_+\Big)
 \\
 &= \mathrm{vol}\{(\vec y,z):  \vec y\in \R^{\nattribs-1}, \; z \geq 0, \;
 |\vec y|^2+(z+x_\nattribs)^2 \leq \rho^2\} 
 \\
 &=\int_{\substack{ \vec y\in \R^{\nattribs-1}\\ |\vec y|^2\leq \rho^2-x_\nattribs^2}} 
 \int_{z=0}^{-x_\nattribs+\sqrt{\rho^2-|\vec y|^2}}dz d\vec y
 \\
 &= \int_{  |\vec y|^2\leq \rho^2-x_\nattribs^2} \Big( \sqrt{\rho^2-|\vec y|^2}-x_\nattribs\Big) d^{\nattribs-1}\vec y
\\
& =  \int_{  |\vec y|^2\leq \rho^2 } \mathbf 1(|\vec y|^2 +x_\nattribs^2\leq \rho^2)  \Big( \sqrt{\rho^2-|\vec y|^2}-x_\nattribs \Big)
 d^{\nattribs-1}\vec y .
\end{split}
\]
Now integrate over $x_\nattribs\in [0,\rho]$: Switching order of integration gives
\[
\begin{split}
\int_0^\rho  A(x_\nattribs)dx_\nattribs &= \int_0^\rho \int_{  |\vec y|^2\leq \rho^2 } \mathbf 1(|\vec y|^2 +x_\nattribs^2\leq \rho^2)  \Big( \sqrt{\rho^2-|\vec y|^2}-x_\nattribs \Big)
 d^{\nattribs-1}\vec y dx_\nattribs 
 \\
 & = \int_{  |\vec y|^2\leq \rho^2 } \int_{x_\nattribs=0}^\rho \Big( \sqrt{\rho^2-|\vec y|^2}-x_\nattribs \Big)    \mathbf 1(|\vec y|^2 +x_\nattribs^2\leq \rho^2)  dx_\nattribs d^{\nattribs-1}\vec y
 \\
 &= \int_{  |\vec y|^2\leq \rho^2 } \int_{x_\nattribs=0}^{\sqrt{\rho^2-|\vec y|^2} }\Big( \sqrt{\rho^2-|\vec y|^2}-x_\nattribs \Big)    dx_\nattribs d^{\nattribs-1}\vec y
 \\
 &= \int_{  |\vec y|^2\leq \rho^2 }(\rho^2-|\vec y|^2)  d^{\nattribs-1}\vec y - \int_{  |\vec y|^2\leq \rho^2 }\frac 12  (\rho^2-|\vec y|^2) d^{\nattribs-1}\vec y
 \\
 &= \frac 12  \int_{  |\vec y|^2\leq \rho^2 }(\rho^2-|\vec y|^2)  d^{\nattribs-1}\vec y
 \\
 &= \rho^{\nattribs+1}  \frac 12  \int_{  |\vec y|  \leq 1 }(1-|\vec y|^2)  d^{\nattribs-1}\vec y .
\end{split}
\]

When $\nattribs=2$, this equals
\[
\rho^3 \frac 12 \int_{y=-1}^1 (1-y^2 )dy =    \frac 23 \rho^3
\]
and dividing by the area of $B(0,\rho) = \pi \rho^2$ gives
\[
\frac 1{\area B(0,\rho)} \int_0^\rho A(x_2)dx_2 = \frac{2}{3\pi}\rho =c_2\rho .
\]

For $\nattribs\geq 3$, we will use polar coordinates: In $G\geq 2$ dimensions  (we will take both $G=\nattribs-1$ and $G=\nattribs$) 
\begin{equation*}
x_j = r \cos(\theta_j) \prod_{k=1}^{j-1} \sin\theta_k, \quad j=1,\dots,G-1, \quad  
x_{G} = r \prod_{k=1}^{G-1} \sin\theta_k
\end{equation*}
with $r\geq 0$, $0\leq\theta_j\leq \pi$ for $j=1,\dots,G-2$ and $0\leq \theta_{G-1} \leq 2\pi$.
The Jacobian of this transformation is
\begin{equation*}
J_{G}(r,\theta) = r^{G -1} \prod_{j=1}^{G-2} (\sin\theta_j)^{G-1-j}.  
\end{equation*}
The volume of the   ball $B(0,\rho)\subset \R^\nattribs$ in dimension $G=\nattribs$ is thus
\[
\begin{split}
\Vol_\nattribs B(0,\rho) &= \int_{r=0}^\rho  r^{\nattribs-1}  dr \prod_{j=1}^{\nattribs-2}\int_{\theta_j=0}^\pi    (\sin\theta_j)^{\nattribs-1-j}   d\theta_j \int_{\theta_{\nattribs-1}=0}^{2\pi}d\theta_{\nattribs-1} 
\\
&= 
\rho^\nattribs \frac { 2\pi }\nattribs \int_{\theta_1=0}^\pi (\sin\theta_1)^{\nattribs-2}  d\theta_1 \cdot 
\prod_{j=2}^{\nattribs-2} \int_{0}^\pi (\sin\theta_j)^{\nattribs-1-j}   d\theta_j   
\\
&= 
\frac { 2\pi }\nattribs \int_{\theta_1=0}^\pi (\sin\theta_1)^{\nattribs-2}  d\theta_1 \cdot 
\prod_{i=1}^{\nattribs-3} \int_{0}^\pi (\sin\theta_j)^{\nattribs-2-i}   d\theta_i    \cdot \rho^\nattribs
\\
&= \frac { 2\pi }\nattribs  \frac{\sqrt{\pi } \Gamma \left(\frac{\nattribs-1}{2}\right)}{\Gamma \left(\frac{\nattribs}{2}\right)}
\cdot \prod_{i=1}^{\nattribs-3} \int_{0}^\pi (\sin\theta_j)^{\nattribs-2-i}   d\theta_i    \cdot \rho^\nattribs .  
\end{split}
\]
 
Using polar coordinates in $\R^{\nattribs-1}$, $\nattribs\geq 3$ (so that $G=\nattribs-1$), gives
\[
\begin{split}
\int_{  |\vec y|  \leq 1 }(1-|\vec y|^2)  d^{\nattribs-1}\vec y 
& = \int_{0}^1(1-r^2) r^{\nattribs-2} dr 
\prod_{j=1}^{\nattribs-3} \int_0^\pi (\sin \theta_j)^{\nattribs-2-j} d\theta_j  \int_0^{2\pi} d\theta_{\nattribs-2} 
\\
&= \frac 2{\nattribs^2-1}  2\pi \prod_{j=1}^{\nattribs-3} \int_0^\pi (\sin \theta_j)^{\nattribs-2-j} d\theta_j
\end{split}
\]
so that 
\[
 \int_{ 0}^\rho A(x_\nattribs) dx_\nattribs  =    \frac 1{\nattribs^2-1}  2\pi \prod_{j=1}^{\nattribs-3} \int_0^\pi (\sin \theta_j)^{\nattribs-2-j} d\theta_j \cdot \rho^{\nattribs+1} . 
 \]
Dividing we find that 
\[
%\begin{split}
\frac 1{ \mathrm{vol}B(0,\rho)} \int_{ 0}^\rho A(x_\nattribs) dx_\nattribs 
 % &
=  \frac{  \nattribs}{\nattribs^2-1}
 \frac{  \Gamma \left(\frac{\nattribs}{2}\right)}{ \sqrt{\pi } \Gamma \left(\frac{\nattribs-1}{2}\right)}\cdot \rho
 %\\ & 
 = \frac{\Gamma(\frac \nattribs 2+1)}{2\sqrt{\pi}\Gamma(\frac{\nattribs+3}2)}\cdot \rho 
 %=c_\nattribs \rho  .
%\end{split}
\]
which equals $c_\nattribs \rho$.
 \end{proof}

We can now complete the proof of  Proposition~\ref{prop type II}: 
Until now, we have fixed the coordinates $(x_1,\dots, x_{\nattribs-1})$, where the particular face of the cell is $(x_1,\dots, x_{\nattribs-1},0)\in H\cap D_j$;  integrating over these coordinates, we obtain the $(\nattribs-1)$-dimensional volume of that face up to an error of $O(\rho^2)$ , and summing over all interior faces of the cell $D_j$  and then over the various cells, we obtain 
\[
\frac 1{\Vol \attrib} \sum_j  \Big( \Vol(x\in D_j: {\rm \; of \; type \; II})- \Vol_{\nattribs-1}(\partial^{\rm int} D_j) c_\nattribs \rho \Big)  +O(\rho^2)
\]
as asserted by Proposition~\ref{prop type II}. \qed

\section{Distance Based Matching Metric}

We explore the sensitivity of our results in Figure 2 of the main text to a matching measure which is based on distance rather than a matching/non-matching binary classification. One can say, that a binary match/non-match classification does not provide information as to how much the chosen alternative is worse than the optimal one, and thus it makes it harder to evaluate the overall dissatisfaction in the population. A metric which measures the average distance between the possible chosen alternatives and the true location, might provide additional information as to the level of satisfaction of the individual from the chosen alternative.  

To construct such metric, consider an individual at location $x$, within a single Voronoi cell $D_j$. If there is no uncertainty in the perceived location of that individual ($\rho=0$),the individual is assigned to alternative $j$, at a distance $d(\rho = 0,x) = |x-j|$ to the alternative. If the individual mistakenly perceives his position as $y$, leading to choosing another alternative, $i$, then the distance between the true location and the chosen alternative is $|x - i|$. Assuming a uniformly distributed error ball of radius $\rho$ around $x$, we obtain the average distance between the chosen alternative and the true position as:

\begin{equation} \label{d_of_rho}
    d \left( \rho, x \right) = \sum_j  \frac{\Vol\Big(D_j\cap B(x,\rho)\Big)}{\Vol B(x,\rho)}\left|  x - j \right|
\end{equation}

Supplementary Figure~\ref{fig:Vornoi_distance}a shows the effect of uncertainty on the average distance to the chosen alternative: $d \left( \rho, x \right) - d \left( \rho = 0 , x \right) $. By integrating over the attribute space we obtain the average distance between all of the individuals and their chosen alternatives. To measure the elasticity of the overall match on the error $\rho$, we divide the above average by the average distance obtained for $\rho=0$:

\begin{equation} \label{ave_d_of_rho}
    \ave{d \left( \rho \right)} = \frac{\int  d \left( \rho, x \right) dx}{\int  d \left( \rho = 0, x \right) dx} 
\end{equation}

The metric $ \ave{d \left( \rho \right)} $ represents the average incremental distance between the true location and all the possible chosen alternatives within the error ball, relative to the no error case. The larger is its deviation from 1, the higher is the distance between the individuals and their chosen alternatives. 

Supplementary Figure~\ref{fig:Vornoi_distance}a visualizes the effect of uncertainty for the case of the attribute space shown in Figures 1 and 2 of the main text. We plot $d(\rho,x) - d(\rho=0,x)$ at each point of the attribute space. Just as in the case of the binary metric, most of the effect of the uncertainty lies within a strip of radius $\rho$ around the boundaries. However, unlike the binary metric, where the boundaries are the most sensitive to the occurrence of a mismatch, for the distance-based metric, the boundaries are the regions which are the least sensitive to a mismatch, as the distance to the alternatives on both sides of the boundary is of similar magnitude.

Supplementary Figure~\ref{fig:Vornoi_distance}b shows $1/ \ave{d \left( \rho \right)} $ vs. $\rho$ for the same attribute space. The value  of $1/ \ave{d \left( \rho \right)} $ decreases with $\rho$. Note, that unlike the linear decrease in the binary metric, this decrease, for low values of $\rho$, can be fitted by a parabola $ \ave{d \left( \rho \right)} \propto \rho^2$. The effect of the uncertainty is thus second order in $\rho$. 

To compare the effect of the uncertainty between the binary and the distance based cases, consider $\rho=0.15$. The mismatch probability in the binary case is 20$\%$  (as shown in Figure 2 of the main text), however the value of $\ave{d(\rho)}$ is 1.03. That is, although on average 20$\%$ of the population is expected to choose an alternative which is not optimal, the distance to their chosen alternatives is expected to increase by 3$\%$.

% This quantity depends strongly on the relative position of x within the Voronoi cell. We normalize $d \left( \rho, x \right)$ and make it dimensionless by dividing it: $\Tilde{d} \left( \rho, x \right) =  d \left( \rho, x \right) / d \left( \rho = 0, x \right)$ . Averaging over the entire attribute space we get:
% \begin{equation}
%    \ave{\Tilde{d} \left( \rho \right)} =\frac 1{\Vol \attrib}  \int \frac{d \left( \rho, x \right)}{d \left( \rho = 0, x \right)} dx
% \end{equation}
% This is the average additional distance, in percentage, resulted by the uncertainty $\rho$.

\begin{figure}[ht]
\begin{center}
\includegraphics[width = \textwidth]{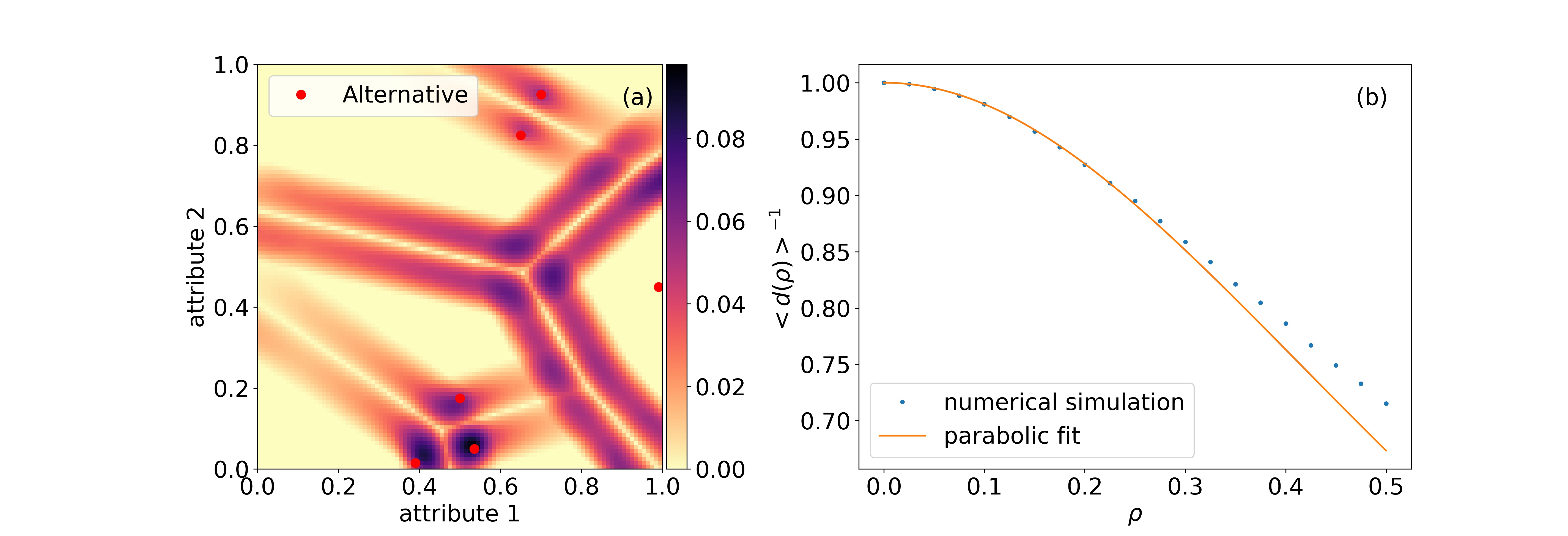}
\caption{\textbf{Distance based matching.}  \textbf{a.} The effect of the uncertainty on the average distance to the chosen alternative: $d(\rho,x) - d(\rho=0,x)$ obtained using numerical simulations for the example shown in Figure 1 of the main text and $\rho = 0.125$ \textbf{b.} $1/\ave{d (\rho)}$ vs. $\rho$ for the same example. 
}
\label{fig:Vornoi_distance}
\end{center}
\end{figure}

\end{document}

\newtheorem{open}{Open Problem}
\newtheorem{thm}{Theorem}[section]
\newtheorem*{thm*}{Theorem}%[section]
\newtheorem{lem}[thm]{Lemma}
\newtheorem{prop}[thm]{Proposition}
\newtheorem{cor}[thm]{Corollary}
\newtheorem*{defn}{Definition}
\newtheorem*{remark}{Remark}
\newtheorem{conj}{Conjecture}%[section]
\newtheorem{exercise}{Exercise}
\newtheorem*{exercise*}{Exercise}
\newtheorem{example}{Example}
\newtheorem*{example*}{Example}

%\numberwithin{equation}{section}

\newcommand{\Z}{{\mathbb Z}} %cph changed from \mathbf
\newcommand{\Q}{{\mathbb Q}}
\newcommand{\R}{{\mathbb R}}
\newcommand{\C}{{\mathbb C}}
\newcommand{\N}{{\mathbb N}}
\newcommand{\FF}{{\mathbb F}}
\newcommand{\fq}{\mathbb{F}_q}
\newcommand{\feq}{\overline{\mathbb F}_q}

\newcommand{\rmk}[1]{\footnote{{\bf Comment:} #1}}

\renewcommand{\mod}{\;\operatorname{mod}}
\newcommand{\ord}{\operatorname{ord}}
\newcommand{\TT}{\mathbb{T}}
\renewcommand{\i}{{\mathrm{i}}}
\renewcommand{\d}{{\mathrm{d}}}
\renewcommand{\^}{\widehat}
\newcommand{\HH}{\mathbb H}
\newcommand{\Vol}{\operatorname{vol}}
\newcommand{\area}{\operatorname{area}}
\newcommand{\tr}{\operatorname{tr}}
\newcommand{\norm}{\mathcal N} % norm =(\frac{ n+\sqrt{n^2-4}} 2)^2
\newcommand{\intinf}{\int_{-\infty}^\infty}
\newcommand{\ave}[1]{\left\langle#1\right\rangle} %  average
\newcommand{\Var}{\operatorname{Var}}
\newcommand{\Prob}{\operatorname{Prob}}
\newcommand{\sym}{\operatorname{Sym}}
\newcommand{\disc}{\operatorname{disc}}
\newcommand{\CA}{{\mathcal C}_A}
\newcommand{\cond}{\operatorname{cond}} % conductor
\newcommand{\lcm}{\operatorname{lcm}}
\newcommand{\Kl}{\operatorname{Kl}} %Kloosterman sum
\newcommand{\leg}[2]{\left( \frac{#1}{#2} \right)}  % Legendre symbol

\newcommand{\sumstar}{\sideset \and^{*} \to \sum}

\newcommand{\LL}{\mathcal L} %L-function of u
\newcommand{\sumf}{\sum^\flat}
\newcommand{\Hgev}{\mathcal H_{2g+2,q}}
\newcommand{\USp}{\operatorname{USp}}
\newcommand{\conv}{*}
\newcommand{\dist} {\operatorname{dist}}
\newcommand{\CF}{c_0} % Fejer constant
\newcommand{\kerp}{\mathcal K}

\newcommand{\fs}{\mathfrak S}
\newcommand{\rest}{\operatorname{Res}} % resultant
\newcommand{\af}{\mathbb A} % affine line
\newcommand{\Li}{\operatorname{Li}}
\newcommand{\Sel}{\mathcal S}
\newcommand{\SF}{\mathbf 1_{\rm SF}}
 \newcommand{\SFz}{\mathbf 1_{{\rm SF},z}}

\newcommand{\T}{\mathbb T} % torus
\newcommand{\E}{\mathbb E}
\newcommand{\length}{\operatorname{length}}
\newcommand{\nattribs}{K} % number of attributes

\title[Supplementary material]{Modeling allocation match under uncertainty: The expected probability of correct assignment
\\ $ \quad$ 
\\
Supplementary material}
\author{Tom Dvir, Renana Peres and Ze\'ev Rudnick}
%\address{HUJI ADDRESS}
%\email{tom.dvir@gmail.com, peresren@huji.ac.il}
%\address{Raymond and Beverly Sackler School of Mathematical Sciences,Tel Aviv University, Tel Aviv 69978, Israel}
%\email{rudnick@tauex.tau.ac.il}

%\thanks{The work was supported  by the European Research Council under the European Union's Seventh
%Framework Programme (FP7/2007-2013)/ERC grant agreement n$^{\text{o}}$~786758 (Z.R.) }

  \date{\today}

\maketitle
 
 We give a proof of the formula for the first variation of the expected probability of correct assignment. 
 
 \tableofcontents
 
\section{A simple formula for the expectation}

We start with a partition of the unit torus $\T^\nattribs =(\Z/\R)^\nattribs$  into $J$ cells with disjoint interiors: $\T^\nattribs= \coprod_{j=1}^J D_j$ . In the application at hand, this is the Voronoi tessellation into basins of attraction. 

We fix an error radius $\rho>0$, and for a point $x\in \T^\nattribs$ we ask what is the probability $P_\rho(x)$ that we select  the correct 
basin of attraction (cell), using an error bar of $\rho$? That is, given that $x\in D_j$, what is the probability  
that we select $D_j$ as the basin of attraction, using an error bar of $\rho$?  Note that the problem only makes sense for $\rho$ small, because once $\rho$ is sufficiently large so that the ball $B(x,\rho)$ covers all of $\T^\nattribs$, say $\rho>\rho_{\max}$, then the question is independent of $\rho$. 
 
We can write a formula for the expected value $\ave{P_\rho} $ of $P_\rho(x)$ (as we average over $x$): 
\begin{prop}
For $\rho<\rho_{\max}$,  
\begin{equation}\label{gen formula}
\ave{P_\rho} =  \int_{\T^\nattribs} \frac{\chi_\rho(z;0)}{\Vol B(0,\rho)}  \sum_j \Vol\Big(D_j\cap(D_j-z)\Big) dz
\end{equation}
where $B(x, \rho)$ is the ball around $x$ of radius $\rho$, and $\chi_\rho(x, y)=1$ if $\dist(x,y)\leq \rho$, and $\chi_\rho(x,y)=0$ otherwise.  

If $\rho>\rho_{\max}$, then $\ave{P_\rho} $ saturates at 
$\ave{P_\rho} = \sum_{j=1}^J \Big( \Vol D_j \Big)^2$. 
\end{prop}
\begin{proof}
To see \eqref{gen formula}, note that  given that $x\in D_j$,   the probability $p(D_j,\rho)(x)$ 
that we select $D_j$ as the basin of attraction  is  the relative area of the intersection of the ``ball'' of radius $\rho$ with the cell $D_j$: 
\[
p(D_j,\rho)(x)=
% \Prob(x\in B(x,\rho)\cap D_j) \Big| x\in D_j) = 
\begin{cases}\frac{\Vol\left(D_j\cap B(x,\rho)\right)}{\Vol B(x,\rho)},&x\in D_j\\   \\
0,&x\notin D_j . \end{cases}
\]

We want to compute the expected value of $p(D_j,\rho)(x)$ (we average over $x$):
\[
\begin{split}
\ave{P_\rho} &=\sum_j \E(p(D_j,\rho)(x)  ) 
%\\&\mbox{where }\E(\bullet |x\in D_j) \mbox{ denotes the conditional expectation}
\\&= 
\sum_j  \int_{D_j} p(D_j,\rho)(x) dx
\\
& =\sum_j  \int_{D_j} \frac{\Vol\Big(D_j\cap B(x,\rho)\Big)}{\Vol B(x,\rho)}dx .
\end{split}
\]
%Let 
%\[
%\chi_\rho(x,y) = \begin{cases} 1,&\dist(x,y)\leq \rho \\0,&{\rm else.} \end{cases}
%\] 
Then
\[
\Vol\Big(D\cap B(x,\rho)\Big) = \int_{\T^\nattribs}\mathbf 1_{D}(y)\chi_\rho(x;y)dy
\]
and 
\[
\begin{split}
 \int_{D} \frac{\Vol\Big(D\cap B(x,\rho)\Big)}{\Vol B(x,\rho)}dx 
 &= 
 \int_{\T^\nattribs} \mathbf 1_{D}(x)\int_{\T^\nattribs} \mathbf 1_{D}(y)\frac{\chi_\rho(x;y)}{\Vol B(0,\rho)} dy dx
\\
&= \int_{\T^\nattribs} \frac{\chi_\rho(z;0)}{\Vol B(0,\rho)} \Big(\int_{\T^\nattribs} \mathbf 1_{D}(x)\mathbf 1_{D}(x+z) dx \Big) dz
\\
&= \int_{\T^\nattribs} \frac{\chi_\rho(z;0)}{\Vol B(0,\rho)} \Vol\Big(D\cap(D-z)\Big) dz .
\end{split}
\]
Thus we find
\[
\ave{P_\rho} =  \int_{\T^\nattribs} \frac{\chi_\rho(z;0)}{\Vol B(0,\rho)}  \sum_j \Vol\Big(D_j\cap(D_j-z)\Big) dz .
\]

To see the saturation value, take $\rho$ to be larger than the diameter of the space of attributes $\TT^\nattribs$, so that for each point $x$, the ball $B(x,\rho)$ coincides with all of the space of attributes $\TT^\nattribs$. %Applying Formula~\eqref{gen formula}
Then for each $x\in \TT^\nattribs$, the intersection $B(x,\rho)\cap D_j=D_j$, and we can compute $\ave{P_\rho}$ simply as a conditional expectation, by writing
\[
P_\rho(x) = \sum_{j=1}^J \mathbf 1_{D_j}(x) \frac{\Vol D_j}{\Vol \TT^\nattribs}
\]
and then 
\[
\begin{split}
\ave{P_\rho} &= \frac 1{\Vol \TT^\nattribs}\int_{\TT^\nattribs}P_\rho(x) dx 
 =  \frac 1{\Vol \TT^\nattribs}\sum_{j=1}^J \int_{\TT^\nattribs} \mathbf 1_{D_j}(x) \frac{\Vol D_j}{\Vol \TT^\nattribs}dx 
\\
&= \sum_{j=1}^J \Big( \frac{\Vol D_j}{\Vol \TT^\nattribs}\Big)^2
= \sum_{j=1}^J \Big( \Vol D_j\Big)^2
\end{split}
\]
as claimed. 
\end{proof}

\section{The one dimensional case}
Formula \eqref{gen formula} makes sense in any dimension,   
but it is only in dimension $\nattribs=1$ that we know how to extract an exact expression from it  (Figure~\ref{fig:expected plots}).    
%So lets turn to the one-dimensional setting $\nattribs=1$.
%Note that in the one-dimensional case, every partition of the circle into intervals with disjoint interiors is in fact a Voronoi tessellation. 
%In this case, we can explicitly compute the expected probability of correct assignment (Figure~\ref{fig:expected plots}). 

 \begin{figure}[ht]
\begin{center}
\includegraphics[height=60mm]{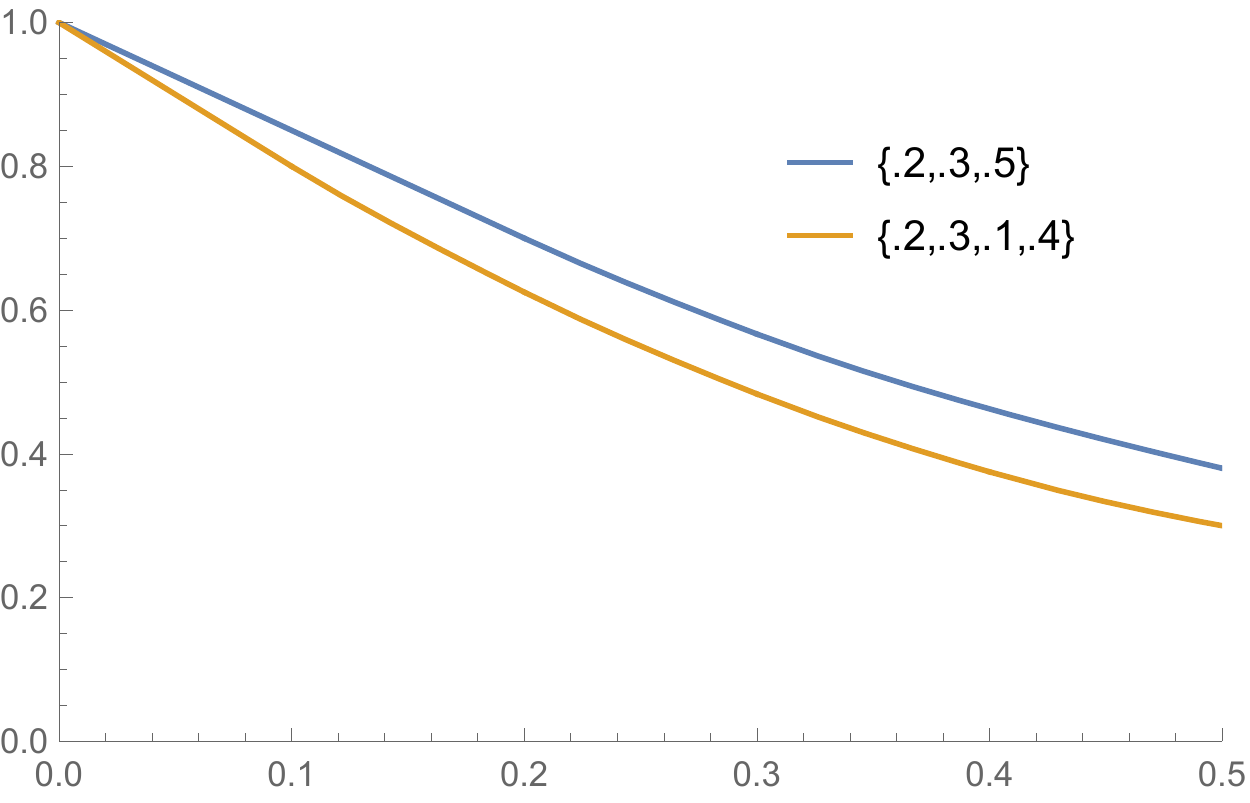}
\caption{ Plots of the expected probability of correct assignment $\ave{P_\rho}$ vs. the error radius $\rho$, for the tessellations $\{.2,.3,.5\}$ (top) and $\{.1,.2,.3,.4\}$ (bottom). }
\label{fig:expected plots}
\end{center}
\end{figure}

\begin{prop}\label{prop:dim 1}
If $\T^1 = \bigcup_{j=1}^J D_j$ is a tessellation of the circle into $J\geq 2$ intervals with disjoint interiors, then for $\rho\leq 1/2$, 
\[
\ave{P_\rho}  
=  \sum_{j:\rho\leq |D_j|}(|D_j|-\frac \rho 2)+\sum_{j: \rho>|D_j|} \frac{(|D_j|)^2}{2\rho} ,   
\]
where $|D|:=\length D$. 

If $\rho\geq \frac 12$, then $\ave{P_\rho}  =  \sum_{j=1}^J  |D_j|^2$. 
\end{prop}

The function $\rho\mapsto \ave{P_\rho}$ is monotonic decreasing, continuous, but not smooth. See Figure~\ref{fig:derivative plot} for the derivative. 

 If $\rho \leq \min(|D_j|)$ then we get 
 \[
 \ave{P_\rho}  =\sum_j (|D_j|-\frac \rho 2) = 1-\frac \rho 2 \cdot J
 \]
 taking into account that $\sum_{j=1}^J |D_j| = \length(\T^1) = 1$. Thus for error radius $\rho$ smaller than the minimal size of the cells, we get a linear function, with slope $-J/2$, where $J$ is the number of cells.  
 
% For $\rho\geq \max (|D_j|)$  we find that 
 %\[\ave{P_\rho}  =(\frac 12 \sum_{j=1}^J  |D_j|^2)\cdot \frac 1{\rho} \]
%and in particular, for the maximal admissible value $\rho=1$ we have $\ave{P_\rho}  =\frac 12 \sum_{j=1}^J  |D_j|^2$. 

 \begin{figure}[ht]
\begin{center}
\includegraphics[height=60mm]{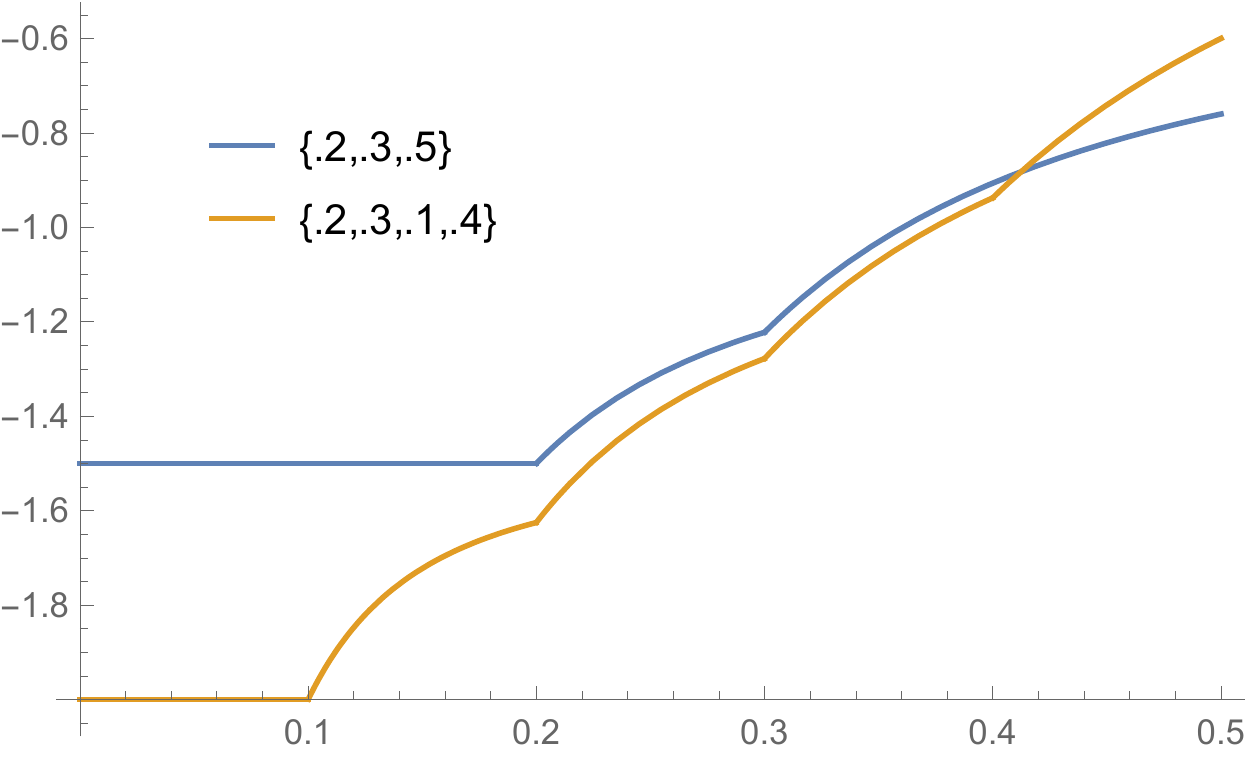}
\caption{Derivative of the expected probability of correct assignment $d\ave{P_\rho}/d\rho$ vs. the error radius $\rho$, for the tessellations $\{.2,.3,.5\}$ (top) and $\{.1,.2,.3,.4\}$ (bottom). }
\label{fig:derivative plot}
\end{center}
\end{figure}

\begin{proof}
A tessellation of the unit circle $\T^1$ is defined by 
\[
\T^1 = \bigcup_{j=1}^J D_j, \quad D_j = (c_j-r_j,c_j+r_j) . 
\]

We can compute the length  of the intersection of an interval and its translate
\[
\length\Big(D_j\cap(D_j-z)\Big) = \max\Big(2r_j-|z|,0 \Big)
\]
and so (using $\length(B(0,\rho))=  2\rho$)
\[ 
\begin{split}
\ave{P_\rho} &=  \sum_j \int_{\T^\nattribs}\frac{\chi_\rho(z;0)}{2\rho}\max\Big(2r_j-|z|,0 \Big) dz 
\\
&=  \frac 1{2\rho}  \sum_j \int_{|z|<\rho} \max\Big(2r_j-|z|,0 \Big) dz .
\end{split}
\]
We have
\[
\int_{|z|<\rho} \max\Big(2r-|z|,0 \Big) dz = \begin{cases} 
\int_{-2r}^{2r} 2r-|z| dz = (2r)^2, &\rho>2r\\
\int_{-\rho}^{\rho} 2r-|z| dz =4r\rho-\rho^2,& \rho\leq 2r 
\end{cases}
\]
which gives
\[
\ave{P_\rho}  = \sum_{j: \rho>2r_j} \frac{(2r_j)^2}{2\rho} + \sum_{j:\rho\leq 2r_j}(2r_j-\frac \rho 2)
= \sum_{j: \rho>|D_j|} \frac{(|D_j|)^2}{2\rho} + \sum_{j:\rho\leq |D_j|}(|D_j|-\frac \rho 2)
\]
 noting that $2r_j = \length D_j =:|D_j|$. 
 
 In particular, if $\rho \leq \min(|D_j|)$ then we get 
 \[
 \ave{P_\rho}  =\sum_j (|D_j|-\frac \rho 2) = 1-\frac \rho 2 \cdot J
 \]
 taking into account that $\sum_{j=1}^J |D_j| = \length(\T^1) = 1$. Thus for error radius $\rho$ smaller than the minimal size of the cells, we get a linear function, with slope $-J/2$, where $J$ is the number of cells.  
 
 For $\rho\geq 1/2$, we have $B(x,\rho)=\T^1$ for any $x$, and then instead of dividing by $2\rho$, we divide by $1$, since $B(x,\rho)=\T^1$ for $\rho\geq 1$. Moreover, for all $j$ we must have $|D_j|\leq 1/2\leq \rho$ (assuming $J\geq 2$), and so we are left with 
 \[
 \ave{P_\rho}  =  \sum_{ j} |D_j|^2
 \]
 as claimed. 
 \end{proof}

 \section{ The first variation of $\ave{P_\rho}$}  
 Our goal in this section is to obtain an exact formula for the {\em first variation} of $\ave{P_\rho}$, that is for the slope at $\rho=0$. 
 
 We denote by 
\[ 
\ave{f(z)}_{z\in B(\rho)}:=\frac 1{\Vol B(\rho)}\int_{B(\rho)} f(z)dz
\]
the average of a function on the ball $B(\rho)=B(0,\rho)$. Let 
\begin{equation}
c_{\nattribs} = \frac 12 \ave{|z_1|}_{z\in B(1)}
\end{equation}
(the average over the unit ball).  We shall show (see Lemma~\ref{lem:ave |z|}) that
\begin{equation}\label{formula for c_K}
c_{\nattribs}  = \frac 12 \frac{\Gamma \left(\frac{\nattribs}{2}+1\right)}{\sqrt{\pi } \Gamma \left(\frac{\nattribs+3}{2}\right)}  =
\begin{cases}
\frac 1{\pi} \frac{2^{2m}}{(m+1)\binom{2m+1}{m}},& \nattribs=2m\;{\rm even} \\
\\
  \frac{1}{2^{2m+2}}\binom{2m+1}{m}, &\nattribs=2m+1\;{\rm odd}. 
\end{cases}
\end{equation}
%Thus $c_1=\frac 14$, $c_2=\frac {2}{3\pi}$, $c_3 = \frac{3}{16}$, etcetera. 
% \begin{center}
% \begin{table}[h] 
% \begin{tabular}{ | l|  c|c|c|c|c|c|c|c|c | r |}
%    \hline
%    $\nattribs$ & $1$ & $2$ & $3$ & $4$ & $5$ &  $6$  & $7$  & $8 $ & $9$ & $ 10$ \\  
%     \hline   
%    \large    $c_\nattribs$ &    \large $\frac 14$ &  \large $\frac {2}{3 \pi}$ &  \large  $\frac {3}{16}$ &  \large  $\frac{8}{15 \pi}$ &  \large$\frac{5}{32}$ &  \large$\frac{16}{35\pi}$ & 
 %    \large$\frac{35}{256}$ &  \large$\frac{128}{315 \pi}$ &  \large$\frac{63}{512}$ & \large $\frac{256}{693 \pi}$ \\
 %   \hline
 %   
 %  \end{tabular}
%    \bigskip
 %  \caption{$c_\nattribs= \frac 12 \frac{\Gamma \left(\frac{\nattribs}{2}+1\right)}{\sqrt{\pi } \Gamma \left(\frac{\nattribs+3}{2}\right)}$} 
%\end{table}
%\end{center}

 \begin{prop}\label{prop:expand near 0} 
In dimension $\nattribs \geq 2$, the mean probability for correct assignment $\ave{P_\rho}$ for $\rho$ small is 
\[
\ave{P_\rho} \sim 1- \left(c_{\nattribs} \sum_j \Vol_{\nattribs-1}(\partial D_j) \right)\cdot \rho  ,\quad \rho \searrow 0 .
\]
\end{prop}
 \begin{proof}
 We write formula \eqref{gen formula} as
 \[
 \ave{P_\rho} =     \sum_j \ave{ \Vol\Big(D_j\cap(D_j-z) \Big) }_{z\in B(\rho)} 
\]
 We now use a geometric fact, that for any a reasonable set $D\subset \T^{\nattribs}$, say bounded with piecewise smooth (or even piecewise linear) boundary, such as a Voronoi cell, we have 
 \[
\ave{ \Vol\Big(D\cap(D-z) \Big) }_{z\in B(\rho)}   \sim \Vol(D) - c_{\nattribs}\Vol_{\nattribs-1}(\partial D) \cdot \rho   ,\quad \rho \searrow 0.
\]  
%Surely this is standard but for now I haven't been able to find a reference, so give an argument in 
We give the proof in Section~\ref{sec:average translate}. 
The Proposition follows by summing over all cells $D_j$, noting that $\sum_j \Vol D_j=\Vol \T^{\nattribs}=1$. 
 \end{proof}
 
% \appendix
 \section{Averaging translates of regions}\label{sec:average translate}

Let $D \subset \R^{\nattribs}$ be a reasonable set, say with piecewise smooth (or even piecewise linear) boundary, such as a Voronoi cell. 
Denote by $B(x, \rho)$ the ball or radius $\rho$ around $x$, and set $B(\rho) = B(0,\rho)$. We want to compute the slope at $\rho=0$ of the function
\[
\frac 1{B(\rho)}\int_{B(\rho)}\Vol\Big (D\cap (D-z) \Big) dz
\]  
%Surely this is standard but for now I haven't been able to find a reference.
\begin{prop}\label{lem:conv}
As $\rho\searrow 0$,
\[
\frac 1{B(\rho)}\int_{B(\rho)}\Vol\Big (D\cap (D-z) \Big) dz \sim  \Vol(D) - c_{\nattribs} \Vol_{\nattribs-1}(\partial D) \cdot \rho  .
\]  
\end{prop}

%\begin{cor}
%In dimension $d=2$, the mean probability for correct assignment $\ave{P_\rho}$ for $\rho$ small is 
%\[\ave{P_\rho} \sim 1- \left(c_2\sum_j \length(\partial D_j) \right)\cdot \rho +O(\rho^2),\quad \rho \searrow 0\]
%where $c_2 =  2/(3\pi)$. 
%\end{cor}
% In the one-dimensional case, we saw that for $\rho\leq \min_j |D_j|$, the function $\ave{P_\rho}$ is linear, of the form
% \[ E(P_\rho) = 1-\frac{J}{2}\rho,\quad \rho\leq \min_j |D_j| . \]
% We can interpret this as $E(P_\rho) =1- c_1\sum_j\#\partial D_j \cdot \rho$, with $c_1=1/4$, which is the same form as Lemma~\ref{lem:conv} once we note that the boundary of an interval consists of $2$ points, so that $\#\partial D_j=2$. 
 
\subsection{Proof of Proposition~\ref{lem:conv} when $\nattribs=2$}
We denote by 
\[ \ave{f(z)}_{z\in B(\rho)}:=\frac 1{\area B(\rho)}\int_{B(\rho)} f(z)dz
\]
the average of a function on the ball $B(\rho)=B(0,\rho)$. We want to show that as $\rho \searrow 0$, 
\[
\frac 1{B(\rho)}\int_{B(\rho)}\area\Big (D\cap (D-z) \Big) dz \sim  \area(D) - \frac {2}{3\pi}\length(\partial D) \cdot \rho . 
\]

We will only prove the claim when $D$ is a finite polygon\footnote{A general region is treated by approximating it by a polygon for which the formula is known, and  and then taking a limit over approximating polygons.  }, 
which is the case of a Voronoi cell.

We take $a=a(\rho)$ such that 
\[
\frac{a}{\rho} \to \infty  \quad {\rm but}\quad \frac{a^2}{\rho}\to 0, \quad {\rm as}\;  \rho \to  0 . 
\]
For instance take $a(\rho) = \rho^{3/4}$. 
We  tile the sides of $D$ by a single layer of small squares $\{S_j\}$ of side-length $2a$, with a small remaining region $D_{\rm rem}$  near the corners of the polygon, of total area  $O(a^2)$, and an interior region $D_{\rm int}$ which is at distance $>a$ from the boundary 
as in Figure~\ref{fig:approximate_cell}, so in particular that  any translate $D_{\rm int}-z$ with $|z|\leq \rho \ll a$ is entirely contained in $D$. 
\begin{figure}[ht]
\begin{center}
\includegraphics[height=60mm]{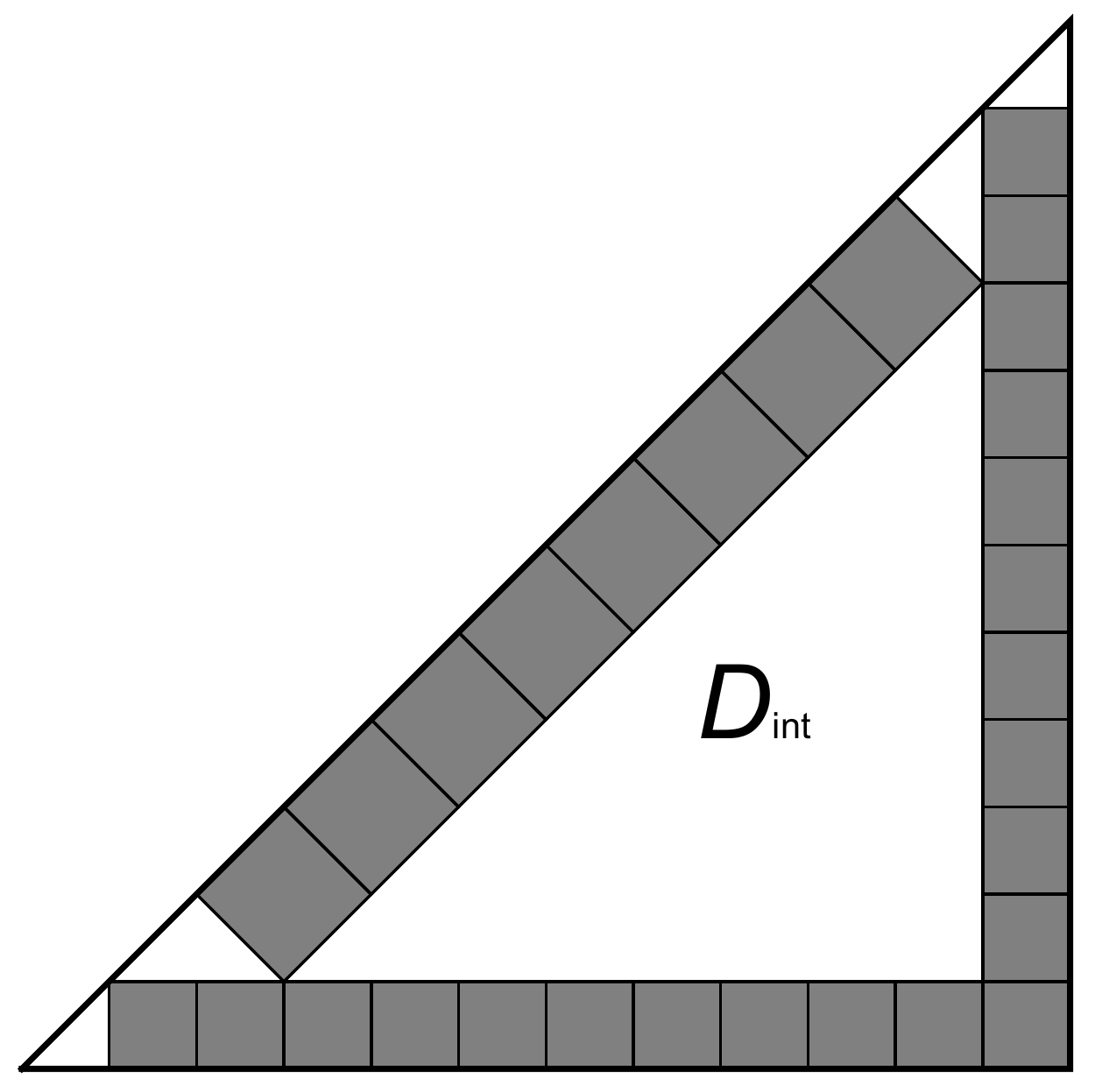}
\caption{  Approximating a right triangle by boundary layer consisting of a union of squares, and an interior part. }
\label{fig:approximate_cell}
\end{center}
\end{figure}

Then the indicator function of $D$ is the sum 
\[ \mathbf 1_D = \mathbf 1_{D_{\rm int}}+\sum_j \mathbf 1_{S_j}+\mathbf 1_{\rm D_{\rm rem}} .
\]
We have
\begin{multline*}
\frac 1{B(\rho)}\int_{B(\rho)}\area\Big (D\cap (D-z) \Big) dz = 
\ave{ \int_{\T^2} \mathbf 1_D(x)\mathbf 1_D(x+z) dx}_{z\in B(\rho)}
\\=
 \ave{ \int_{\T^2} \mathbf 1_D(x)\mathbf 1_{D_{\rm int}}(x+z) dx}_{z\in B(\rho)}
 + \ave{ \int_{\T^2} \mathbf 1_D(x)\mathbf 1_{D_{\rm rem}}(x+z) dx}_{z\in B(\rho)}
\\ + 
\sum_j\ave{ \int_{\T^2} \mathbf 1_D(x)\mathbf 1_{S_j}(x+z) dx}_{z\in B(\rho)} .
\end{multline*}

The translate $D_{\rm int}\subset D$ is entirely contained in $D$ for $|z|\leq \rho$, and so   
$ \mathbf 1_D(x)\mathbf 1_{D_{\rm int}}(x+z) = 1$ for $x\in D$, and then 
\begin{equation}\label{interior region}
\ave{ \int_{\T^2} \mathbf 1_D(x)\mathbf 1_{D_{\rm int}}(x+z) dx}_{z\in B(\rho)} = \area(D_{\rm int}) .
\end{equation}

For the term with the small remaining region, we bound
\[
 \int_{\T^2} \mathbf 1_D(x)\mathbf 1_{D_{\rm rem}}(x+z) dx\leq  \area(D_{\rm rem})   =O(a^2) 
 \]
 and the same holds for the average over $B(\rho)$. Since we assume $a^2/\rho\to 0$, this term is negligible for our purposes. 
  
Next we study the contributions from a  square $S_0$ which tiles the boundary of $D$. All but finitely many (depending on the number of sides of the polygon) have two neighbouring squares $S_j\subset D$, $j=1,2$ having a common side with $S_0$ (Figure~\ref{fig:type3}). We may assume that the boundary is the $x$-axis, and the rectangles $S_j$ lie above it, and the interior region $D_{\rm int}$ lies above the rectangles. 
\begin{figure}[ht]
\begin{center}
\includegraphics[height=20mm]{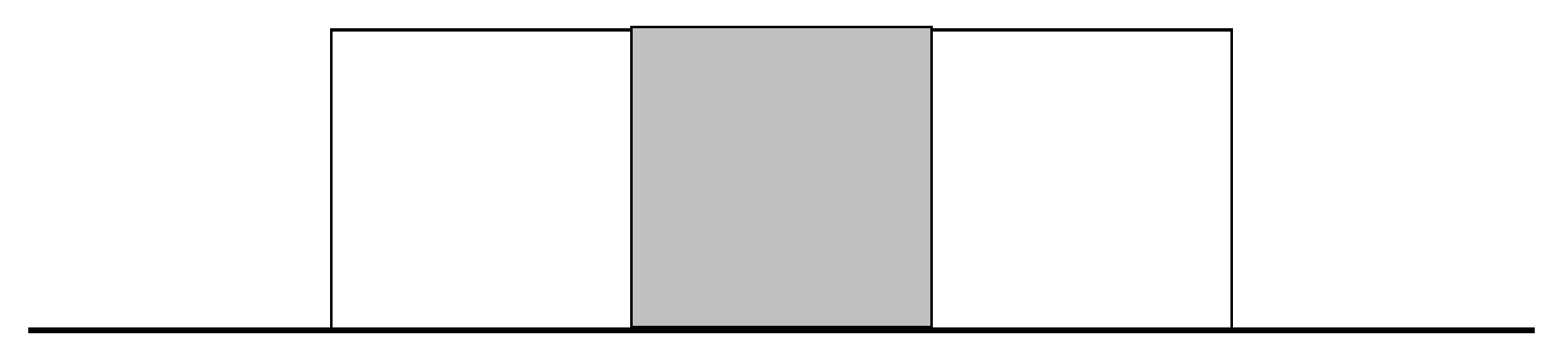}
\caption{  A  boundary square $S_0$ (gray) with two neighbouring squares. }
\label{fig:type3}
\end{center}
\end{figure}

For such a square $S_0$, 
\begin{multline*}
\ave{ \int_{\T^2} \mathbf 1_D(x)\mathbf 1_{S_0}(x+z) dx}_{z\in B(\rho)} =
\ave{ \int_{\T^2} \mathbf 1_{S_0}(x)\mathbf 1_{S_0}(x+z) dx}_{z\in B(\rho)}
\\
+\sum_{j=1}^2 
\ave{ \int_{\T^2} \mathbf 1_{S_j}(x)\mathbf 1_{S_0}(x+z) dx}_{z\in B(\rho)}
\\+ \ave{ \int_{\T^2} \mathbf 1_{D_{\rm int}}(x)\mathbf 1_{S_0}(x+z) dx}_{z\in B(\rho)} .
\end{multline*}
%We assume $\rho$ is much smaller than the side-length of the squares, and then the inner integral  vanishes unless $S_j$ either shares a side or a corner with $S_0$, or if $S_j=S_0$. 

The intersection of a square $S=[-a,a]\times [-a,a]$ and of its translate $S+w$ is a rectangle of side lengths $2a-|w_1|$ and $2a-|w_2|$, if $\max(|w_1|,|w_2|)\leq 2a$, and is empty otherwise. Hence the area of the intersection is 
\[
\begin{split}
\area\Big(S\cap (S-w)\Big) &= (2a-|w_1|)(2a-|w_2|) 
\\&= \area(S)-2a(|w_1|+|w_2|)+|w_1 w_2| .
\end{split}
\]

Hence for $\rho\ll a$, 
\[
\ave{ \int_{\T^2}  \mathbf 1_{S_0}(x)\mathbf 1_{S_0}(x+z) dx}_{z\in B(\rho)} =\ave{ \area(S_0)-2a(|z_1|+|z_2|)+|z_1 z_2| }_{z\in B(\rho)}  . 
 \]
 Now
 \[
 \ave{|z_1|}_{z\in B(\rho)}  = \frac 1{\pi\rho^2} \int_0^{2\pi}\int_0^\rho r|\cos\theta| d\theta rdr =\frac{4}{3\pi}\rho
 \]
 and likewise for $ \ave{|z_2|}_{z\in B(\rho)}$. Further, 
 \[
  \ave{|z_1z_2|}_{z\in B(\rho)} =  \frac 1{\pi\rho^2} \int_0^{2\pi}\int_0^\rho r^2|\cos\theta \sin \theta| d\theta rdr  = \frac{\rho^2}{2\pi} =O(\rho^2)
 \]
 so that 
 \[
 \begin{split}
\ave{ \int_{\T^2}  \mathbf 1_{S_0}(x)\mathbf 1_{S_0}(x+z) dx}_{z\in B(\rho)} 
& -\area(S_0) 
 = -\frac{16a}{3\pi}\rho+O(\rho^2) \\&= -\frac{2}{3\pi}\length(\partial S_0) \cdot \rho +O(\rho^2) . 
 \end{split}
 \]
 
Next we compute the overlap with one of the two neighboring squares $S_j$, sharing a common side (assumed to be the $x$-axis),  
 say $S_j=S_0 + (2a,0)$.  
 %(squares sharing a common vertex will contribute $O(\rho^2)$).  
 The intersection of $S_0$ and $S_j-z = S_0+(2a-z_1,-z_2)$ has area 
\[
(2a-|z_2|)\cdot |z_1|  , \qquad {\rm if} \quad z_1>0
\]
(assuming $|z_1|, |z_2|<2a$), and the intersection is empty if $z_1<0$, so that for $z_1>0$, $\max(|z_1|,|z_2|)<2a$ 
\[
\int_{\T^2}  \mathbf 1_{S_0}(x)\mathbf 1_{S_0+(2a,0)}(x+z) dx  =(2a-|z_2|)\cdot  z_1 ,\quad z_1>0  .
\]
 Since  $\rho\ll a$, the average over the ball $B(0,\rho)$ is 
\[
 \begin{split}
 &\ave{ \int_{\T^2}  \int_{\T^2} 
   \mathbf 1_{S_0}(x)\mathbf 1_{S_0+(2a,0)}(x+z) dx}_{z\in B(\rho)}   =
 \ave{ (2a-|z_2|)\cdot  z_1 \cdot \mathbf 1_{z_1>0}}_{z\in B(\rho)}
 \\
 & =
  \frac 1{\pi \rho^2}\int_0^\rho \int_{-\pi/2}^{\pi/2} \Big( 2a r\cos\theta   -
r|\sin\theta| \cdot r\cos \theta \Big)      \;  rdr d\theta    
    \\
   &= \frac{4a}{3\pi}\rho+O(\rho^2) .
 \end{split}
 \]
 
 Finally, the overlap of $S_0-z$ with $D_{\rm int}$ is a rectangle of side lengths $2a$ and $|z_2|$, if $z_2>0$, and is empty if $z_2<0$ (recall we assume that $D_{\rm int}$ lies above the rectangles). Hence for $\rho<a$, the average over the ball $B(0,\rho)$ is
  \begin{equation*}
 \ave{ \area(D_{\rm int}\cap (S-z) )}_{z\in B(\rho)}  = \ave{2a z_2 \cdot \mathbf 1_{z_2>0} }_{z\in B(\rho)} 
     =    \frac{4a}{3\pi}\rho .
 \end{equation*}

  We thus see that the contribution of each such square is
\begin{equation}\label{contribution per square}
\begin{split}
\ave{ \int_{\T^2} \mathbf 1_D(x)\mathbf 1_{S_0}(x+z) dx}_{z\in B(\rho)} -\area(S_0) 
&=
  -\frac{16a}{3\pi}\rho+3  \frac{4a}{3\pi}\rho+O(\rho^2)
  \\
  &= -\frac{4a}{3\pi}\rho+O(\rho^2) . 
 \end{split}
 \end{equation}

We sum the RHS of \eqref{contribution per square} over all boundary squares and note  that each such square contributes $2a$ (one side) 
to the length of the boundary, and the corresponding sides cover all of the boundary except for that part 
of the boundary lying in the remaining region $D_{\rm rem}$, which has length $O(a)=o(\rho^{1/2})$, so that 
\[ 
\sum_{S_j \;{\rm boundary}} -\frac{4a}{3\pi}\rho
=-\frac{2}{3\pi}\Big( \length \partial D+O(a)\Big) \rho \sim -\frac{2}{3\pi}  \length \partial D \cdot \rho .
\]
This gives 
\begin{multline*}
\sum_{S_j \;{\rm boundary}}  \ave{ \int_{\T^2} \mathbf 1_D(x)\mathbf 1_{S_j}(x+z) dx}_{z\in B(\rho)}  
\\
\sim  \sum_{S_j \;{\rm boundary}} \area(S_j)-\frac{2}{3\pi}  \length \partial D \cdot \rho    .
\end{multline*}
Recalling that the interior region squares contributes exactly its area by \eqref{interior region}, and the remaining region gives a contribution of $O(a^2)=o(\rho)$   gives
\[
  \ave{ \int_{\T^2} \mathbf 1_D(x)\mathbf 1_D(x+z) dx}_{z\in B(\rho)}  
\sim \area(D) -\frac{2}{3\pi} \length (\partial D) \cdot \rho
\]
as claimed. \qed
     
  \subsection{Proof of Proposition~\ref{lem:conv} in general dimension $\nattribs$}
Let $D\subset \R^\nattribs$ be an arbitrary polytope (an intersection of hyperplanes). As in the two-dimensional case, we tile the sides of $D$ by a single layer of small $\nattribs$-dimensional cubes $\{S_j\}$ of side-length $2a$, chosen so that 
\[
\rho = o(a), \qquad a^{ 2} = o(\rho), 
\]
with a small remaining region $D_{\rm rem}$  near the codimension-two faces of the polygon, of total volume  $O(a^2)$, 
and an interior region $D_{\rm int}$ which is at distance $>a$ from the boundary %as in Figure~\ref{fig:approximate_cell}, 
 so in particular that  any translate $D_{\rm int}-z$ with $|z|\leq \rho \ll a$ is entirely contained in $D$.  
 Then the indicator function of $D$ is the sum 
\[ \mathbf 1_D = \mathbf 1_{D_{\rm int}}+\sum_j \mathbf 1_{S_j}+\mathbf 1_{\rm D_{\rm rem}} .
\]

We have
\begin{multline*}
\frac 1{B(\rho)}\int_{B(\rho)}\Vol\Big (D\cap (D-z) \Big) dz = 
\ave{ \int_{\T^\nattribs} \mathbf 1_D(x)\mathbf 1_D(x+z) dx}_{z\in B(\rho)}
\\=
 \ave{ \int_{\T^\nattribs} \mathbf 1_D(x)\mathbf 1_{D_{\rm int}}(x+z) dx}_{z\in B(\rho)}
 + \ave{ \int_{\T^\nattribs} \mathbf 1_D(x)\mathbf 1_{D_{\rm rem}}(x+z) dx}_{z\in B(\rho)}
\\ + 
\sum_j\ave{ \int_{\T^\nattribs} \mathbf 1_D(x)\mathbf 1_{S_j}(x+z) dx}_{z\in B(\rho)} .
\end{multline*}

The translate $D_{\rm int}\subset D$ is entirely contained in $D$ for $|z|\leq \rho$, and so   
$ \mathbf 1_D(x)\mathbf 1_{D_{\rm int}}(x+z) = 1$ for $x\in D$, and then 
\begin{equation}
\label{interior region}
\ave{ \int_{\T^2} \mathbf 1_D(x)\mathbf 1_{D_{\rm int}}(x+z) dx}_{z\in B(\rho)} = \Vol(D_{\rm int}) .
\end{equation}

For the term with the small remaining region, we bound
\[
 \int_{\T^\nattribs} \mathbf 1_D(x)\mathbf 1_{D_{\rm rem}}(x+z) dx\leq  \Vol(D_{\rm rem})   =O(a^2) 
 \]
 and the same holds for the average over $B(\rho)$. Since we assume $a^2/\rho\to 0$, this term is negligible for our purposes. 
  
  Next we study the contributions from a  cube $S_0$ which tiles the boundary of $D$. Recall that a $\nattribs$-dimensional cube has $2\nattribs$ faces of codimension one  (that is of dimension $\nattribs-1$). All but a small number of them  (depending on the number of sides of the polytope) have $2(\nattribs-1)$ neighbouring cubes $S_j\subset D$, $j=1,\dots,2(\nattribs-1)$ having a common $(\nattribs-1)$-dimensional face with $S_0$. We may assume that the relevant part of the boundary is the plane $x_\nattribs=0$, and the cubes $S_j$ lie above it, and the interior region $D_{\rm int}$ lies above the cubes. 

For such a cube $S_0$, 
\begin{multline*}
\ave{ \int_{\T^\nattribs} \mathbf 1_D(x)\mathbf 1_{S_0}(x+z) dx}_{z\in B(\rho)} =
\ave{ \int_{\T^\nattribs} \mathbf 1_{S_0}(x)\mathbf 1_{S_0}(x+z) dx}_{z\in B(\rho)}
\\
+\sum_{j=1}^{2(\nattribs-1)} 
\ave{ \int_{\T^\nattribs} \mathbf 1_{S_j}(x)\mathbf 1_{S_0}(x+z) dx}_{z\in B(\rho)}
\\+ \ave{ \int_{\T^\nattribs} \mathbf 1_{D_{\rm int}}(x)\mathbf 1_{S_0}(x+z) dx}_{z\in B(\rho)} .
\end{multline*}

The intersection of a cube $S=[-a,a]^\nattribs$ and of its translate $S+w$ is a cube  of side lengths $2a-|w_j|$ ($j=1,\dots, \nattribs$), if 
$\max_j(|w_j|)\leq 2a$, and is empty otherwise. Hence the volume of the intersection is 
\[
\Vol(S\cap (S-w)) = \prod_{j=1}^\nattribs (2a-|w_j|) = (2a)^\nattribs- (2a)^{\nattribs-1} \sum_{j=1}^\nattribs |w_j|+\dots 
\]
Therefore for $\rho\ll a$, 
\begin{multline*}
\ave{ \int_{\T^\nattribs}  \mathbf 1_{S_0}(x)\mathbf 1_{S_0}(x+z) dx}_{z\in B(\rho)} 
\\
=\ave{ (2a)^\nattribs-(2a)^{\nattribs-1} \sum_{j=1}^\nattribs |z_j|+\dots }_{z\in B(\rho)} \\
= (2a)^\nattribs-\nattribs(2a)^{\nattribs-1}\ave{ |z_1| }_{z\in B(\rho)} +O(\rho^2) .
\end{multline*}

 Next we compute the overlap with one of the $2(\nattribs-1)$ neighboring cubes $S_j$, sharing a common codimension-one face 
 %(assumed to be the $x$-axis),  
 say $S_j=S_0 + (2a,0,0,\dots,0)$.  
 %(squares sharing a common vertex will contribute $O(\rho^2)$).  
 The intersection of $S_0$ and $S_j-z = S_0+(2a-z_1,-z_2,\dots,-z_\nattribs)$ has volume 
\[
 |z_1|\cdot\prod_{j=2}^\nattribs (2a-|z_j|)  , \qquad {\rm if} \quad z_1>0
\]
(assuming $|z_j|<2a$), and the intersection is empty if $z_1<0$, so that for $z_1>0$, $\max(|z_j|)<2a$ 
\[
\int_{\T^\nattribs}  \mathbf 1_{S_0}(x)\mathbf 1_{S_0+(2a,\vec 0)}(x+z) dx  =z_1 \cdot \prod_{j=2}^\nattribs (2a-|z_j|) = (2a)^{\nattribs-1} \cdot z_1 +O(\rho^2) .
\]
  Since  $\rho\ll a$, the average over the ball $B(0,\rho)$ is 
  \begin{multline*}
  \ave{ \int_{\T^\nattribs}  
  \int_{\T^\nattribs}  \mathbf 1_{S_0}(x)\mathbf 1_{S_0+(2a,\vec 0)}(x+z) dx}_{z\in B(\rho)}  
  \\
   =(2a)^{\nattribs-1}   \ave{z_1 \mathbf 1_{z_1>0}}_{z\in B(\rho)} +O(\rho^2)
  \\
  = (2a)^{\nattribs-1} \frac 12  \ave{|z_1|}_{z\in B(\rho)}+O(\rho^2) .
 \end{multline*}
 
 Finally, the overlap of $S_0-z$ with $D_{\rm int}$ is a box   with base a $(\nattribs-1)$-dimensional cube of side length $2a$ and height $|z_\nattribs|$, if $z_\nattribs>0$, and is empty if $z_\nattribs<0$ (recall we assume that $D_{\rm int}$ lies above the cubes). Hence for $\rho<a$, the average over the ball $B(0,\rho)$ is
\[
    \begin{split}
 \ave{ \Vol(D_{\rm int}\cap (S-z) )}_{z\in B(\rho)}  &= \ave{(2a)^{\nattribs-1} z_\nattribs \cdot \mathbf 1_{z_\nattribs>0} }_{z\in B(\rho)} 
 \\&     = (2a)^{\nattribs-1} \frac 12  \ave{|z_1|}_{z\in B(\rho)} .     
 \end{split}
  \]
     
  Thus each boundary cube gives a contribution of 
  \begin{multline*}
%  \begin{split}
  \ave{ \int_{\T^\nattribs}  \mathbf 1_D(x)\mathbf 1_{S_0}(x+z) dx}_{z\in B(\rho)}  =
  \\
   (2a)^\nattribs -\nattribs(2a)^{\nattribs-1}\ave{ |z_1| }_{z\in B(\rho)} 
  \\ +2(\nattribs-1) (2a)^{\nattribs-1} \frac 12  \ave{|z_1|}_{z\in B(\rho)}
\\ 
+ (2a)^{\nattribs-1} \frac 12  \ave{|z_1|}_{z\in B(\rho)}  + O(\rho^2)
  \\
  \sim (2a)^\nattribs -  \frac 12 (2a)^{\nattribs-1}\ave{|z_1|}_{z\in B(\rho)}  
  \\
  = (2a)^\nattribs -  \frac 12 (2a)^{\nattribs-1}\ave{|z_1|}_{z\in B(1)}  \cdot \rho
  %\end{split}
  \end{multline*}  
  after noting that a change of variables $z=\rho x$, $|x|\leq 1$ gives
\[
\ave{ |z_1| }_{z\in B(\rho)} =\rho \cdot \ave{ |z_1| }_{z\in B(1)} .
\]

  We sum   over all boundary cubes, and note  that each such cube contributes $(2a)^{\nattribs-1}$ (one side) 
to the length of the boundary, and the corresponding sides cover all of the boundary except for that part 
of the boundary lying in the remaining region $D_{\rm rem}$, which has negligible $(\nattribs-1)$-dimensional volume, so that 
\[ 
\sum_{S_j \;{\rm boundary}}\sim \sum_{S_j \;{\rm boundary}} \Vol_\nattribs(S_j) - \frac 12 \Vol_{\nattribs-1}(\partial D)\ave{|z_1|}_{z\in B(1)} \cdot \rho  .
\]

Recalling that the interior region squares contributes exactly its volume, and that the remaining region gives a negligible contribution, 
we obtain 
\[
\begin{split}
  \ave{ \int_{\T^\nattribs} \mathbf 1_D(x)\mathbf 1_D(x+z) dx}_{z\in B(\rho)}  
&\sim \Vol(D) - \frac 12 \ave{|z_1|}_{z\in B(1)} \Vol_{\nattribs-1} (\partial D) \cdot \rho
\\&=\Vol(D) - c_\nattribs \Vol_{\nattribs-1} (\partial D) \cdot \rho
  \end{split}
  \]
\qed

\subsection{Computing $c_{\nattribs}$}
Finally, we compute the value of $c_{\nattribs}$:

\begin{lem}\label{lem:ave |z|} 
\[
\ave{ |z_1| }_{z\in B(1)}  =   
\begin{cases}
\frac 1{\pi} \frac{2^{2m+1}}{(m+1)\binom{2m+1}{m}},& \nattribs=2m\;{\rm even} \\
\\
  \frac{1}{2^{2m+1}}\binom{2m+1}{m}, &\nattribs=2m+1\;{\rm odd}.
\end{cases}
\]
\end{lem}
\begin{proof}
%We change variables $z=\rho x$, $|x|\leq 1$, to find 
%\[
%\ave{ |z_1| }_{z\in B(\rho)} =\rho \cdot \ave{ |z_1| }_{z\in B(1)} .
%\]
We use polar coordinates  in $\nattribs$ dimensions 
\begin{equation*}
x_j = r \cos(\theta_j) \prod_{k=1}^{j-1} \sin\theta_k, \quad j=1,\dots,\nattribs-1
\end{equation*}
and
\begin{equation*}
x_{\nattribs} = r \prod_{k=1}^{\nattribs-1} \sin\theta_k
\end{equation*}
with $0\leq r\leq 1$, $0\leq\theta_j\leq \pi$ for $j=1,\dots,\nattribs-2$ and $0\leq \theta_{\nattribs-1} \leq 2\pi$.
The Jacobian of this transformation is
\begin{equation*}
J_{\nattribs}(r,\theta) = r^{\nattribs-1} \prod_{j=1}^{\nattribs-2} (\sin\theta_j)^{\nattribs-1-j}.  
\end{equation*}
The volume of the unit ball $B(1)$ is thus
\[
\begin{split}
\Vol B(1) &= \int_{r=0}^1  r^{\nattribs-1}  dr \prod_{j=1}^{\nattribs-2}\int_{\theta_j=0}^\pi    (\sin\theta_j)^{\nattribs-1-j}   d\theta_j \int_{\theta_{\nattribs-1}=0}^{2\pi}d\theta_{\nattribs-1} 
\\
&= 
\frac { 2\pi }\nattribs \int_{\theta_1=0}^\pi (\sin\theta_1)^{\nattribs-2}  d\theta_1 \cdot 
\prod_{j=2}^{\nattribs-2} \int_{0}^\pi (\sin\theta_j)^{\nattribs-1-j}   d\theta_j   .
\end{split}
\]
The integral of $|x_1|=r|\cos \theta_1|$ is 
\begin{multline*}
\int_{B(1)}|x_1|dx  =
 \int_{r=0}^1   \int_{\theta_{ 1}=0}^{\pi}  r|\cos \theta_1| (\sin\theta_1)^{\nattribs-2}  d\theta_1  r^{\nattribs-1}   dr 
 \\
  \qquad \times \prod_{j=2}^{\nattribs-2}\int_{0 }^\pi  (\sin \theta_j)^{\nattribs-1-j}   d\theta_j  \int_{ 0}^{2\pi}d\theta_{\nattribs-1} 
\\
 = \frac{2\pi }{\nattribs+1} \int_{\theta_1=0}^\pi |\cos \theta_1| (\sin\theta_1)^{\nattribs-2} d\theta_1 \cdot  \prod_{j=2}^{\nattribs-2} 
\int_0^\pi (\sin\theta_j)^{\nattribs-1-j}   d\theta_j  .
\end{multline*}
The average $\ave{ |z_1| }_{z\in B(1)} $, which is the ratio of these two, is therefore
\[
\ave{ |z_1| }_{z\in B(1)} = \frac { \frac{2\pi }{\nattribs+1} \int_{ 0}^\pi |\cos \theta_1| (\sin\theta_1)^{\nattribs-2} d\theta_1 }
{\frac { 2\pi }\nattribs \int_{ 0}^\pi (\sin \theta_1)^{\nattribs-2} d\theta_1 }
=\frac{\Gamma \left(\frac{\nattribs}{2}+1\right)}{\sqrt{\pi } \Gamma \left(\frac{\nattribs+3}{2}\right)} .
\]

It remains to note that  
\[
\frac{\Gamma \left(\frac{\nattribs}{2}+1\right)}{\sqrt{\pi } \Gamma \left(\frac{\nattribs+3}{2}\right)} =
\begin{cases}
\frac 1{\pi} \frac{2^{2m+1}}{(m+1)\binom{2m+1}{m}},& \nattribs=2m\;{\rm even} \\
\\
  \frac{1}{2^{2m+1}}\binom{2m+1}{m}, &\nattribs=2m+1\;{\rm odd} .
\end{cases}
\]
 \end{proof}

%\begin{thebibliography}{99}
%\bibitem{Gilbert}
%E. N. Gilbert. Random Subdivisions of Space into Crystals. 
%Ann. Math. Statist. Volume 33, Number 3 (1962), 958--972.
%
%\bibitem{Meijering}
%J. L. Meijering, Interface area, edge length, and number of vertices in crystal aggregates with random nucleation, Philips Res. Rep., 8, 270--290 (1953). 
%\end{thebibliography}